\documentclass[sigconf]{aamas} 
\usepackage{balance} % for balancing columns on the final page

\usepackage{algorithm}
\usepackage{algorithmic}
\usepackage{bm}
\usepackage{caption}
\usepackage{subcaption}
\usepackage{graphicx}

\renewcommand{\algorithmicrequire}{\textbf{Input:}}

\usepackage{ulem}
\usepackage{amsmath}
\usepackage{hyperref}
\newtheorem{thm}{Theorem}[section]

\newtheorem{df}[thm]{Definition}

\newtheorem{rem}[thm]{Remark}
\newtheorem{prop}[thm]{Proposition}

\usepackage[utf8]{inputenc}
\usepackage[english]{babel}
\usepackage{graphicx}
\usepackage{subcaption}
\usepackage{amssymb}
\usepackage{balance} % for balancing columns on the final page
\usepackage{xcolor}
\usepackage{hyperref}
\usepackage{etoolbox}
\setcopyright{none}
\acmConference[AAMAS '22]{Proc.\@ of the 21st International Conference
on Autonomous Agents and Multiagent Systems (AAMAS 2022)}{May 9--13, 2022}
{Online}{P.~Faliszewski, V.~Mascardi, C.~Pelachaud,
M.E.~Taylor (eds.)}
\copyrightyear{2022}
\acmYear{2022}
\acmDOI{}
\acmPrice{}
\acmISBN{}

\acmSubmissionID{196}

\title[AAMAS-2022 Formatting Instructions]{Agent-Temporal Attention for Reward Redistribution in\\ Episodic Multi-Agent Reinforcement Learning}

\author{Baicen Xiao}
\affiliation{
  \institution{University of Washington}
  \city{Seattle, WA}
  \country{USA}
}
\email{bcxiao@uw.edu}

\author{Bhaskar Ramasubramanian}
\affiliation{
  \institution{Western Washington University}
  \city{Bellingham, WA}
  \country{USA}
}
\email{ramasub@wwu.edu}

\author{Radha Poovendran}
\affiliation{
  \institution{University of Washington}
  \city{Seattle, WA}
  \country{USA}
}
\email{rp3@uw.edu}

\begin{abstract}
This paper considers %the multi-agent temporal credit assignment problem for 
multi-agent reinforcement learning (MARL) tasks where agents receive a shared global reward at the end of an episode. 
The delayed nature of this reward affects the ability of the agents to assess the quality of their actions at intermediate time-steps. 
This paper focuses on developing methods to learn a temporal redistribution of the episodic reward to obtain a dense reward signal. 
Solving such MARL problems requires addressing two challenges:  
identifying (1) relative importance of states along the length of an episode (along time), and (2) relative importance of individual agents' states at any single time-step (among agents). 
In this paper, we introduce  \textbf{Agent-Temporal Attention for Reward Redistribution in Episodic Multi-Agent Reinforcement Learning (AREL)} to address these two challenges. 
\emph{AREL} uses attention mechanisms to characterize the influence of actions on state transitions along trajectories (\emph{temporal attention}), and how %any single 
each agent is affected by other agents at each time-step (\emph{agent attention}). 
The redistributed rewards predicted by \emph{AREL} are dense, and can be integrated with any given MARL algorithm. 
We evaluate \emph{AREL} on challenging tasks from the Particle World environment %, and three maps in 
and the StarCraft Multi-Agent Challenge. 
\emph{AREL} results in higher rewards in Particle World, and improved win rates in StarCraft compared to three state-of-the-art reward redistribution methods. 
Our code is available at \textbf{\url{https://github.com/baicenxiao/AREL}}.
\end{abstract}


\newcommand{\BibTeX}{\rm B\kern-.05em{\sc i\kern-.025em b}\kern-.08em\TeX}

\begin{document}

%%% The following commands remove the headers in your paper. For final 
%%% papers, these will be inserted during the pagination process.

\pagestyle{fancy}
\fancyhead{}

%%% The next command prints the information defined in the preamble.

\maketitle

\section{Introduction}
Cooperative multi-agent reinforcement learning (MARL) involves multiple autonomous agents that learn to collaborate to complete tasks in a shared environment by maximizing a global reward \citep{busoniu2008comprehensive}.
Examples of %multi-agent 
systems where MARL has been used include autonomous vehicle coordination \cite{sallab2017deep}, %multi-player 
and video games \citep{samvelyan19smac, tampuu2017multiagent}.%, and analysis of social dilemmas \cite{leibo2017multi}. 

One approach to enable better coordination %among agents 
is to use a single centralized controller that can access observations of all agents \cite{jiang2018learning}. 
In this setting, algorithms designed for single-agent RL can be used for the multi-agent case.   
However, this may not be feasible when trained agents are deployed independently or when communication costs between agents and the controller is prohibitive. 
In such a situation, agents will need to be able to learn decentralized policies. 

The centralized training with decentralized execution (CTDE) paradigm, introduced in \cite{lowe2017multi, rashid2018qmix}, enables agents to learn decentralized policies efficiently. 
Agents using CTDE can communicate with each other during training, but are required to make decisions independently at test-time. 
The absence of a centralized controller will require each agent to
%In the absence of a centralized controller, each agent will be required to 
assess how its own actions can contribute to a shared global reward. 
This is called the \emph{multi-agent credit assignment problem}, and has  
%Addressing this problem has 
been the focus of recent work in MARL, such as COMA \cite{foerster2018counterfactual}, QMIX \cite{rashid2018qmix} and QTRAN \cite{son2019qtran}. %Actor-attention-critic \cite{iqbal2019actor}. 
Solving the multi-agent credit assignment problem alone, however, is not adequate to efficiently learn agent policies when the (global) reward signal is delayed until the end of an episode.

In reinforcement learning, agents seek to solve a sequential decision problem guided by reward signals at intermediate time-steps. 
This is called the \emph{temporal credit assignment problem} \cite{sutton2018reinforcement}. 
%The agent learns to complete a task by allocating a higher credit to better states and actions. 
In many applications, rewards %signals 
may be delayed. %, thereby making it difficult to ascertain the relative importance of states and actions at each time-step. 
For example, in molecular design %problems 
\cite{olivecrona2017molecular}, \emph{Go} \cite{silver2016mastering}, and computer games such as \emph{Skiing} \cite{bellemare2013arcade}, a summarized score is revealed only at the end of an episode. 
The \emph{episodic} reward implies absence of feedback on quality of actions at intermediate time steps, making it difficult %for RL algorithms 
to learn good policies. 
%Delayed rewards have been shown to introduce a large bias \cite{arjona2019rudder} or variance \cite{ng1999policy} in the performance of RL algorithms. 
The long-term temporal credit assignment problem has been studied in single-agent RL by performing return decomposition via contribution analysis \cite{arjona2019rudder} and using sequence modeling \cite{liu2019sequence}. 
These methods do not directly scale well to MARL since size of the joint observation space grows exponentially with number of agents \cite{lowe2017multi}.

%In reinforcement learning (RL), agents seek to solve a sequential decision problem guided by reward signals at intermediate time-steps. 
%This is called the \emph{temporal credit assignment problem} \cite{sutton2018reinforcement}. 
%%The agent learns to complete a task by allocating a higher credit to better states and actions. 
%In many applications, a reward signal may be delayed. %, thereby making it difficult to ascertain the relative importance of states and actions at each time-step. 
%For example, in molecular design problems \cite{olivecrona2017molecular} and computer games such as \emph{Skiing} \cite{bellemare2013arcade} and \emph{Go} \cite{silver2016mastering}, a summarized score is revealed only at the end of an episode. 
%The \emph{episodic} reward implies the absence of feedback on the quality of actions at intermediate time steps, which makes it difficult %for RL algorithms 
%to learn good policies. 
%%Delayed rewards have been shown to introduce a large bias \cite{arjona2019rudder} or variance \cite{ng1999policy} in the performance of RL algorithms. 
%The long-term temporal credit assignment problem has been studied in single-agent RL environments by performing return decomposition via contribution analysis \cite{arjona2019rudder} and using sequence modeling \cite{liu2019sequence}. 
%These methods do not scale well to multiple agents since the size of the joint observation space grows exponentially with number of agents \cite{lowe2017multi}.
%
Besides scalability, addressing temporal credit assignment in MARL with episodic rewards presents two challenges. %that need to be addressed in tandem. 
It is critical to identify the relative importance of: 
%In order to solve the temporal credit assignment problem in multi-agent environments with episodic rewards, it is critical to identify the relative importance of: 
i) each agent's state at any single time-step (\emph{agent dimension});  
ii) states along the length of an episode (\emph{temporal dimension}). 
We introduce \textbf{Agent-Temporal Attention for Reward Redistribution in Episodic Multi-Agent Reinforcement Learning (AREL)} to address these challenges.% in tandem. 

\emph{AREL} uses attention mechanisms \cite{vaswani2017attention} to carry out \emph{multi-agent temporal credit assignment} by concatenating:  % by learning the importance of each agent's observations relative to that of other agents, and the importance of states at each time step relative to states at other time steps. 
%The attention mechanism in \emph{ATAC} consists of a concatenation of: 
i) a \emph{temporal attention module} to characterize the influence of actions on state transitions along trajectories, and; ii) an \emph{agent attention module}, to determine how any single agent is affected by other agents at each time-step.
The attention modules enable learning a redistribution of the episodic reward %decomposition of the episodic reward in order to redistribute it effectively 
along the length of the episode, resulting in a \emph{dense} reward signal. 
%Rather than concatenating agent observations, \emph{AREL} overcomes the challenge of scalability by analyzing observations of each agent using a temporal attention module that is shared among agents, and whose output is passed to an agent-attention module. 
%To overcome the challenge of scalability, instead of concatenating agents' observations, \emph{AREL} analyzes observations of each agent using a temporal attention module that is shared among agents, and whose output is passed to an agent-attention module.
To overcome the challenge of scalability, instead of working with the concatenation of (joint) agents' observations, \emph{AREL} analyzes observations of each agent using a temporal attention module that is shared among agents. 
The outcome of the temporal attention module is passed to an agent attention module that characterizes the relative contribution of each agent to the shared global reward. 
The output of the agent attention module is then used to learn the redistributed rewards. 
%The number of possible redistributions increases with episode length. 
%To encourage learning simpler redistributions, \emph{AREL} uses a variance-based regularization based on the principle of maximum entropy \cite{cover1999elements}.

When rewards are delayed or episodic, it is important to identify `critical' states that contribute to the reward. 
The authors of \cite{gangwani2020learning} recently demonstrated that rewards delayed by a long time-interval make it difficult for temporal-difference (TD) learning methods to carry out temporal credit assignment effectively. 
AREL overcomes this shortcoming by using attention mechanisms to effectively learn a redistribution of an episodic reward. 
%In comparison, by design, attention mechanisms in AREL are able to effectively learn a redistribution of an episodic reward. 
This is accomplished by identifying critical states through capturing long-term dependencies between states and the episodic reward.
%Since the number of possible decompositions will increase with episode length, we use a variance-based regularization stemming from the principle of maximum entropy \cite{cover1999elements} to encourage learning decompositions that are robust to over-fitting. 

Agents that have identical action and observation spaces are said to be \emph{homogeneous}. % in their contributions. 
%As an illustration, c
Consider a task where two homogeneous agents need to collaborate to open a door by locating two buttons and pressing them simultaneously. 
In this example, while locations of the two buttons (states) are important, the identities of the agent at each button are not. 
This property is termed \emph{permutation invariance}, and can be utilized to make the credit assignment process sample efficient \cite{gangwani2020learning, liu2019sequence}. 
Thus, a redistributed reward must identify whether an agent is in a `good' state, and should also be invariant to the identity of the agent in that state. 
\emph{AREL} enforces this property by designing the credit assignment network with permutation-invariant operations among homogeneous agents, and can be integrated with MARL algorithms to learn agent policies.

We evaluate \emph{AREL} on three tasks from the Particle World environment \cite{lowe2017multi}, and three combat scenarios in the StarCraft Multi-Agent Challenge \cite{samvelyan19smac}.  
In each case, 
agents receive a summarized reward only at the end of an episode. 
We compare \emph{AREL} with three state-of-the-art reward redistribution techniques, and observe that
%Our results indicate that \emph{AREL} accelerates learning of policies, and agents obtain higher rewards %compared to three state-of-the-art methods 
%in Particle World. 
%Agents equipped with \emph{AREL} achieve better win rates in StarCraft.
\emph{AREL} results in accelerated learning of policies and higher rewards in Particle World, and improved win rates in StarCraft. % compared to three state-of-the-art methods. 

%We note that our experiments require computational resources to train the attention modules of \emph{AREL} in addition to those needed to train deep RL algorithms. 
%Using these resources might result in higher energy consumption, especially as the number of agents grows. 
%This is a potential limitation of the methods studied in this paper. 
%However, we believe that \emph{AREL} partially addresses this concern by sharing certain modules among all agents in order to improve scalability.

\section{Related Work}
%
%An RL %reinforcement learning 
%agent aims to determine a sequence of actions to maximize a long-term (cumulative) reward. 
%%that leads to the maximization of a long-term reward. 
%This is termed the \emph{credit assignment problem} \cite{sutton2018reinforcement}. 
%However, sparse or delayed reward signals can affect the performance of RL algorithms \cite{liu2019sequence}. 
%This section provides an overview of related work on credit assignment in single and multi-agent RL. 
%Credit assignment techniques can be distinguished based on 
%the availability of information about the agent's environment. 
%whether knowledge of the agent's environment is available or not. 
%In multi-agent environments, these methods have primarily focused on ascertaining the contribution of individual agents to a shared global reward at each time-step. 
%
Several techniques have been proposed to address temporal credit assignment when prior knowledge of the problem domain is available. 
Potential-based reward shaping is one such method that provided theoretical guarantees in single \cite{ng1999policy} and multi-agent \cite{devlin2011theoretical, lu2011policy} RL, and %environments with sparse rewards. 
%This 
was %empirically 
shown to accelerate learning of policies in \cite{devlin2011empirical}. 
Credit assignment was also studied by incorporating human feedback through imitation learning \cite{kelly2019hg, ross2011reduction} and demonstrations \cite{brown2019extrapolating, huang2020ma}.%, taylor2011integrating, wang2017improving}. 
%These approaches, though, required prior knowledge of the problem domain in order to work effectively. 

When prior knowledge of the problem domain is not available, recent work has studied temporal credit assignment in single-agent RL with delayed rewards. %when rewards are delayed in single-agent RL. % in the single-agent case.  
An approach named RUDDER \cite{arjona2019rudder} used contribution analysis %presented in  
to decompose episodic rewards by computing the difference between predicted returns at successive time-steps. % using `contribution analysis' \cite{arjona2019rudder}. 
%RUDDER computes the difference between predicted returns at successive time-steps and uses this as the redistributed reward. % in this work.
Parallelly, the authors of \cite{liu2019sequence} proposed using natural language processing models for carrying out temporal credit assignment for episodic rewards. 
%Parallelly, in another work, temporal credit assignment for episodic rewards %in single agent RL 
%using natural language processing models was proposed by  \cite{liu2019sequence}. 
%The above methods could be applied in multi-agent scenarios by concatenating observations of individual agents.
The scalability of the above methods to MARL, though, can be a challenge due to the exponential growth in the size of the joint observation space \cite{lowe2017multi}.
%While the above methods perform well in the single-agent case, their scalability in MARL can be a challenge due to the exponential growth in the size of the joint observation space.
%The above methods, though, were not analyzed in scenarios with more than one agent.

In the multi-agent setting, recent work has studied performing multi-agent credit assignment at each time-step. 
Difference rewards were 
%In the multi-agent setting, difference rewards have been 
used to assess the contribution of an agent to a global reward in \cite{agogino2006quicr, devlin2014potential, foerster2018counterfactual} 
%This was accomplished 
by computing a counterfactual term that marginalized out actions of that agent while keeping actions of other agents fixed. 
%Difference and potential rewards were combined in \cite{devlin2014potential} to define a counterfactual term that was used to assign credit. 
%$Q-$learning was combined with difference rewards to carry out multi-agent credit assignment along the length of a trajectory in discrete-state discrete-action environments in \cite{agogino2006quicr}. 
%Difference rewards were used in \cite{foerster2018counterfactual} to define counter-factual multi-agent policy gradients in deep cooperative MARL. 
%
%The authors of \cite{sunehag2018value} proposed value decomposition networks that 
Value decomposition networks, proposed in \cite{sunehag2018value}, 
decomposed a centralized value into a sum of %individual 
agent values to assess each one's contributions. 
A monotonicity assumption on value functions 
was imposed in QMIX \cite{rashid2018qmix} to assign credit to individual agents. A 
%A 
generalized approach to decompose a joint value into individual agent values was presented in QTRAN \cite{son2019qtran}. %, and 
The Shapley Q-value was used in \cite{wang2020shapley} to distribute a global reward to identify each agent's contribution. %in contrast to the shared reward approach. 
The authors of \cite{yang2020q} %proposed a method %Q-value path decomposition 
%to decompose 
decomposed global Q-values along trajectory paths, %to perform credit assignment, 
while \cite{zhou2020learning} used an entropy-regularized method to encourage exploration to aid multi-agent credit assignment. 
%An adaptive entropy-regularized method presented in \cite{zhou2020learning} encouraged consistent exploration to aid multi-agent credit assignment. 
%These methods were evaluated on tasks where rewards were dense, and actions were discrete. 
The above techniques %were developed to identify contributions of individual agents to a dense global reward at each time-step, 
did not address long-term temporal credit assignment and hence will not be adequate for learning policies efficiently when rewards are delayed.
%Since these methods did not carry out a temporal redistribution of the reward, they may not be amenable for use in situations with episodic rewards.
%Moreover, these methods did not carry out a temporal redistribution of the reward, and hence cannot be used in situations with episodic rewards. 

Attention mechanisms have been used for multi-agent credit assignment in recent work.
The authors of \cite{mao2019modelling} used an attention mechanism with a CTDE-based algorithm to enable each agent effectively model policies of other agents (from its own perspective). 
%proposed ATT-MADDPG in which the centralized critic is enhanced by attention mechanisms such that the teammates' policies can be modelled in an effective way.
Hierarchical graph attention networks proposed in \cite{ryu2020multi} modeled hierarchical relationships among agents and used two attention networks to effectively represent individual and group level interactions.
%Attention mechanism was incorporated with centrally computed critics in \cite{iqbal2019actor} to select relevant information for each agent to pay attention to.
The authors of \cite{jiang2019graph, liu2020multi} combined attention networks with graph-based representations to indicate the presence and importance of interactions between any two agents. 
%modeled the relationship between agents by a complete graph and adopt an attention network to learn if an interaction exists between two agents.
%Similarly, \cite{jiang2019graph} also utilized graph models and proposed to apply multi-head attention as the convolution kernal and used graph convolution to extract the relation representation between nodes where the encoding of local observation of agent is the feature of node.
The above approaches used attention mechanisms primarily to identify relationships between agents at a specific time-step. 
They did not consider long-term temporal dependencies, and therefore may not be sufficient to learn policies effectively when rewards are delayed. 
%only considered applying attention mechanism to extract agents' relation at a specific time and did not considered agents' long-term dependencies and therefore will not be enough for efficient policy learning if rewards are delayed. 

A method for temporal redistribution of episodic rewards in single and multi-agent RL was recently presented in \cite{gangwani2020learning}. 
%Recently, \cite{gangwani2020learning} presented a method for temporal redistribution of episodic rewards in single and multi-agent RL. 
A `surrogate objective' was used to uniformly redistribute an episodic reward along %to states and actions in 
a trajectory. 
However, this work did not use information from sample trajectories to characterize the relative contributions of agents at intermediate time-steps along an episode.
%identify and distinguish between the quality of states of the agents. 

Our approach differs from the above-mentioned related work in that it uses attention mechanisms for multi-agent temporal credit assignment. % in environments with episodic rewards. 
\emph{AREL} overcomes the challenge of scalability by analyzing observations of each agent using temporal and agent attention modules, which 
%Temporal and agent attention modules 
respectively characterize the effect of actions on state transitions along a trajectory and how each agent is influenced by other agents at each time-step. 
%A temporal attention module characterizes the influence of actions on state transitions along a trajectory, and an agent attention module determines how each agent is influenced by other agents at each time-step. 
Together, these modules will enable an effective redistribution of an episodic reward. % along the length of the episode. 
\emph{AREL} does not require human intervention to guide agent behaviors, and can be integrated with MARL algorithms to learn decentralized agent policies in environments with episodic rewards. 

\section{Background}
\begin{figure*}[!htp]
	\centering
	\includegraphics[width=5.75 in]{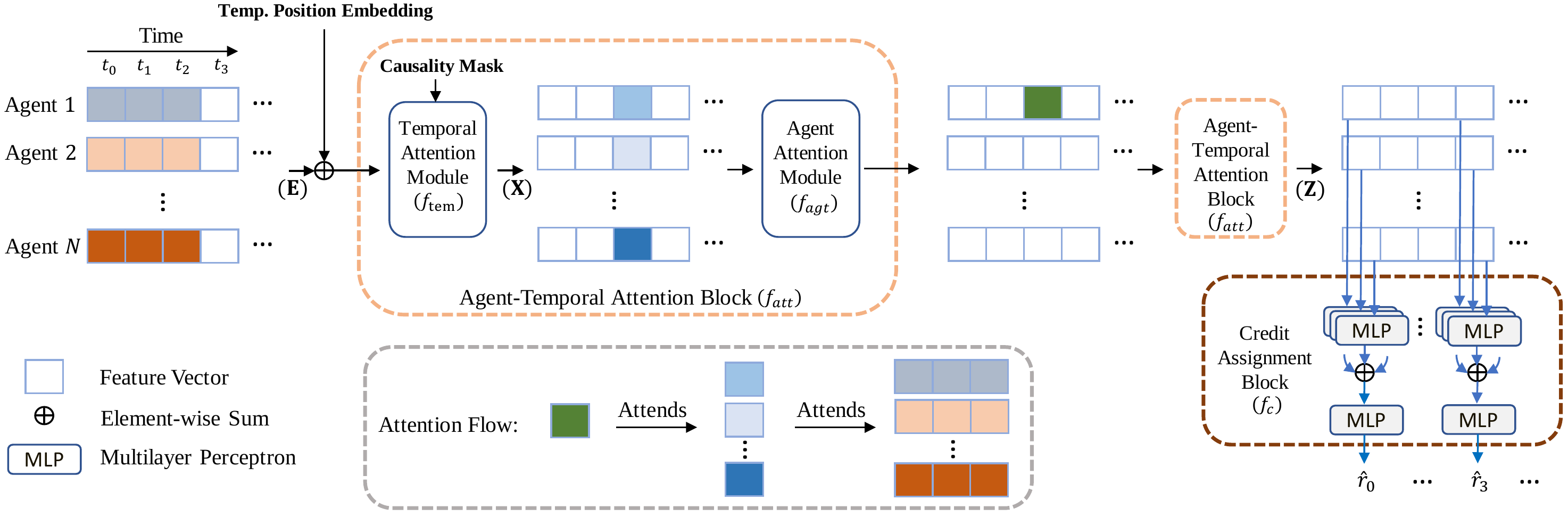} 
	\caption{Schematic of AREL. 
	The agent-temporal attention block concatenates temporal and agent attention modules, and summarizes input feature (e.g. observation) vectors. This is accomplished by establishing relationships between (\emph{attending to}) information along time and among agents. 
	The attention flow indicates that an output feature vector of the agent-temporal attention block for an agent at a time $t$ (green square) can attend to input features from all other agents before and including time $t$. Multiple agent-temporal attention blocks can be concatenated to each other to improve expressivity. 
	The output of the last such block is fed to the credit assignment block, which applies shared multi-layer perceptrons to each attention feature. The output is the redistributed reward, which is integrated with MARL algorithms (e.g. MADDPG, QMIX) to learn agent policies. 
%	It has two main components: 1) Agent-Temporal attention block, which summarize input feature vectors (e.g., observation vectors) by attending information along temporal dimension and agent dimension. 
%	Attention flow shows that each output feature vector at time $t$ of the agent-temporal attention block can indirectly attend input feature from all agents up to time $t$.
%	; and 2) credit assignment block, which applies a single shared multilayer perceptron on each attention feature vector and outputs predicted reward at each time step. 
}\label{AttnNetwSchematic}
\end{figure*}
%
%This sections provides some of the needed formalisms for our framework.
%\subsection{Decentralized Partially Observable MDPs and Episodic MARL}

A fully cooperative multi-agent task can be specified as a decentralized partially observable Markov decision process (Dec-POMDP) \cite{oliehoek2016concise}. 
%Adopting notation from \cite{oliehoek2016concise}, a 
A Dec-POMDP is a tuple $G = (S, A, P, r, Z, O, n, \gamma)$, where  
$s \in S$ describes the environment state. 
Each agent $i \in \{1,2,\dots,n\}$ receives an observation $o^i \in O^i$ according to an observation function $Z(s,i): S \times \mathbb{N} \rightarrow O$. 
At each time step, agent $i$ chooses action $a^i \in A^i$ according to its policy $\pi^i : O^i \times A^i \rightarrow [0,1]$. 
$A^1 \times \dots \times A^n := A$ forms the joint action space, and the environment transitions to the next state according to the function $P: S \times A^1 \times \dots \times A^n \rightarrow S$. 
All agents share a global reward $r: S \times A \rightarrow \mathbb{R}$. 
The goal of the agents is to determine their individual policies to maximize the %total expected 
\emph{return}, %given by 
$J:= \mathbb{E}_{s \sim P, a^1 \sim \pi^1, \dots, a^n \sim \pi^n} [\sum_{t=0}^T \gamma^t r_t(s, a^1,\dots,a^n)]$, where $\gamma$ is a discount factor, and $T$ is the length of the horizon. 
Let $a_t:=(a^1_t,\dots,a^n_t)$ %, $\pi:=(\pi^1,\dots,\pi^n)$, 
and $R_t:=\sum_{l=0}^{T-t} \gamma^l r_{t+1}$. 
A trajectory of length $T$ is an alternating sequence of observations and actions, $\tau:=(o_0, a_0,o_1,a_1,\dots,o_T)$. 

In a typical MARL task, agents receive reward $r(s,a)$ immediately following execution of action $a$ %=(a^1,\dots,a^n)$ 
at state $s$. 
The expected return can then be determined by accumulating rewards at each time step. % in the usual way. 
%policy gradients estimated following Eqn. (\ref{PolGradMulti}) can have large variance. 
%This is because most policy gradient methods are tailor-made for the setting where rewards are \emph{dense}. 
%These dense rewards will be able to provide the necessary information for policy improvement and value-function estimation at each time-step. 
%
In episodic RL, a reward is revealed %to the agents 
only at the end of an episode at time $T$, and agents do not receive a reward at intermediate time-steps. 
As a consequence, the expected return for all $t<T$ will be the same (when $\gamma = 1)$. Therefore, the quality of information available for learning policies will be poor at all intermediate time steps. 
%However, when the reward is only observed at the end of an episode, 
Moreover, delayed rewards have been shown to introduce a large bias \cite{arjona2019rudder} or variance \cite{ng1999policy} in the performance of RL algorithms.
%
%\subsection{Centralized Training with Decentralized Execution}

%Agents will need to learn decentralized policies when dimensions of state and action spaces are large, or when agents are unable to communicate with each other. %, they will need to learn decentralized policies. % using only their individual observations. 
The CTDE paradigm \cite{foerster2018counterfactual, lowe2017multi} can be adopted to learn decentralized policies effectively when dimensions of state and action spaces are large. 
%Decentralized policies can be learned in a centralized manner by adopting the CTDE paradigm \cite{foerster2018counterfactual, lowe2017multi}. 
During training, an agent can make use of information about other agents’ states and actions to aid its own learning. 
At test-time, decentralized policies are executed. 
This paradigm has been used to successfully complete tasks in complex MARL environments \cite{gupta2017cooperative, iqbal2019actor, rashid2018qmix}.%, sunehag2018value, yang2020q}. 

\section{Approach}
%\begin{figure*}[!htp]
%	\centering
%	\includegraphics[width=4.95 in]{figures/attention_schematic.pdf} 
%	\caption{Schematic of AREL. 
%	The agent-temporal attention block concatenates temporal and agent attention modules, and summarizes input feature (e.g. observation) vectors. This is accomplished by establishing relationships between (\emph{attending to}) information along time and among agents. 
%	The attention flow indicates that an output feature vector of the agent-temporal attention block for an agent at a time $t$ (green square) can attend to input features from all other agents before and including time $t$. Multiple agent-temporal attention blocks can be concatenated to each other to improve expressivity. 
%	The output of the last such block is fed to the credit assignment block, which applies shared multi-layer perceptrons to each attention feature. The output is the redistributed reward, which is integrated with MARL algorithms (e.g. MADDPG, QMIX) to learn agent policies. 
%%	It has two main components: 1) Agent-Temporal attention block, which summarize input feature vectors (e.g., observation vectors) by attending information along temporal dimension and agent dimension. 
%%	Attention flow shows that each output feature vector at time $t$ of the agent-temporal attention block can indirectly attend input feature from all agents up to time $t$.
%%	; and 2) credit assignment block, which applies a single shared multilayer perceptron on each attention feature vector and outputs predicted reward at each time step. 
%}\label{AttnNetwSchematic}
%\end{figure*}

This paper considers MARL tasks where agents share the same global reward, which 
%However, this reward 
is received only at the end of an episode. 
The objective is to redistribute this episodic reward for effective multi-agent temporal credit assignment. 
To accomplish this goal, it is critical to identify the relative importance of: i) individual agents’ observations at each time-step, and; ii) observations  along the length of a trajectory. 
%The objective is to learn a redistribution of the episodic reward along the length of the episode, and to identify the contribution of each agent to the global reward at each time-step, and use the redistributed rewards with MARL algorithms to learn agent policies.
%Moreover, the learned redistributed rewards must be suitable for integration with any MARL algorithm to learn agent policies. 
%When there are multiple agents, individual agents may contribute differently to the shared global reward. 
%
%Our goal is to devise a method under the following conditions: 
%(1) a model of the environment is not available, and rewards are provided only at the end of a training episode; 
%(2) %we want to 
%learn a decomposition of the reward along the length of the episode; % in order to compute policy gradients; 
%(3) %we are also interested in determining 
%determine the contribution of each agent to the global reward at each time step; 
%(4) the learned decomposed rewards must enable integration with any MARL algorithm to learn agent policies. 
We introduce \emph{AREL} to address the above challenges. 
%We introduce \emph{AREL} to address the challenge of multi-agent temporal credit assignment in tasks with episodic rewards. 
\emph{AREL} uses an agent-temporal attention block to infer relationships among states at different times, and among agents. 
A schematic is shown in Fig. \ref{AttnNetwSchematic}, and we describe its key components and overall workflow in the remainder of this section. 
%We describe the details of \emph{ATAC} in this section. 
%
\subsection{Agent-Temporal Attention}

In order to redistribute an episodic reward %in a multi-agent setup 
in a meaningful way, we need to be able to extract useful information from trajectories. Each trajectory contains a sequence of observations involving all agents. 
%Consider an episode of length $T$. 
At each time-step of an episode of length $T$, a feature of dimension $D$ corresponds to the embedding of a single observation. 
When there are $N$ agents, %we denote a trajectory from an episode 
a trajectory is denoted by 
$\mathbf{E} \in \mathbb{R}^{T \times N \times D}$. 
The objective is to learn a mapping $f_{arel}(\mathbf{E}): \mathbb{R}^{T \times N \times D} \rightarrow \mathbb{R}^{T}$ to assign credit to the agents at each time-step. 
%In the following, we describe how to construct and learn this mapping. 
%
The information in a trajectory $\mathbf{E}$ comprises two parts: 
(1) \emph{temporal information} between (embeddings of) observations at different time steps: this provides insight into the influence of actions on transitions between states; 
(2) \emph{structural information}: this provides insight into how any single agent is affected by other agents.
%\begin{enumerate}
%\item temporal information between (embeddings of) observations at different time steps: this can provide insight into the influence of actions on transitions of the state; 
%\item structural information: this can provide insight into how any single agent is affected by other agents. 
%\end{enumerate}

These two parts are coupled, and hence studied together. 
The process of learning these relationships is termed \emph{attention}. 
We propose an \emph{agent-temporal attention structure}, inspired by the Transformer \cite{vaswani2017attention}. % to make full use of these two types of information. 
This structure selectively pays attention to different types of information- either from individual agents, or at different time-steps along a trajectory. 
This is accomplished by associating a %higher or lower 
weight to an observation %in the trajectory 
based on its relative importance to other observations along the trajectory. 
The agent-temporal attention structure is formed by concatenating one agent attention module with one temporal attention module. 
The temporal attention modules determine how entries of $\mathbf{E}$ at different time-steps are related (along the `first' dimension of $\mathbf{E}$). 
The agent attention module determines how agents influence one another (along the `second' dimension of $\mathbf{E}$).
%The agent attention module uses the input embeddings at a single time-step to calculate how agents influence one another. 
%For each agent, the temporal attention module determines how the input at different time-steps are `related' (\textbf{use of word INPUT needs to be discussed}). 
%
\subsubsection{\textbf{Temporal-Attention Module}}
The input is a trajectory $\mathbf{E} \in \mathbb{R}^{T \times N \times D}$. 
To calculate the temporal attention feature, we obtain the transpose of $\mathbf{E}$ as $\bar{\mathbf{E}} \in \mathbb{R}^{N \times T \times D}$.
Adopting notation from \cite{vaswani2017attention}, each row $\mathbf{e}_i \in \mathbb{R}^{T \times D}$ of $\bar{\mathbf{E}} $ is transformed to a \emph{query} $Q^{tem}_i:=\mathbf{e}_iW^{tem}_q$, \emph{key} $K^{tem}_i:=\mathbf{e}_iW^{tem}_k$, and \emph{value} $V^{tem}_i: =\mathbf{e}_iW^{tem}_v$. 
%Here 
$W^{tem}_q, W^{tem}_k, W^{tem}_v \in \mathbb{R}^{D \times D}$ are learnable parameters, and $i \in \{0,\dots, N-1\}$.
The $t^{th}$ row $x_{i,t}$ of the temporal attention feature $\mathbf{x}_i \in \mathbb{R}^{T \times D}$ is a weighted sum %of values, % at each time step: 
$x_{i,t} := \alpha _{i,t}^T V^{tem}_i$. 
%\begin{align}
%	x_{i,t} = \alpha _{i,t}^T V^{tem}_i.
%\end{align}
%
The attention weight vector $\alpha_{i,t} \in \mathbb{R}^{T\times 1}$ is %computed as 
a normalization (\emph{softmax}) of the inner-product between the $t^{th}$ row of $Q^{tem}_i$, $q^{tem}_{i,t}$, and the key matrix $K^{tem}_i$: 
\begin{align}
	\alpha_{i,t} ^T = \text{softmax} (\frac{q^{tem}_{i,t} {K^{tem}_i}^T}{\sqrt{D}}\odot m^T_t), \label{TempAttnVect}
\end{align}

\noindent
where $\odot$ is an element-wise product, and $m_t$ is a \emph{mask} with its first $t$ entries equal to $1$, and remaining entries $0$. %set to $0$. 
%As a result, at any time $t$, information beyond this time will not be used to assign credit. 
%This ensures causality.  
The mask preserves causality by ensuring that at any time $t$, information beyond $t$ will not be used to assign credit. 
%The mask is needed to ensure causality, so that when assigning credit at a time $t$, we only utilize information available up to this time. 
%
%A temporal position embedding method \cite{devlin2018bert} is used to maintain information about the relative position of a state in an episode. 
A temporal positional embedding \cite{devlin2018bert} maintains information about relative positions of states in an episode. 
Position embeddings are learnable vectors associated to each temporal position of %the length of 
a trajectory. 
The sum of position and trajectory embeddings forms the input to the temporal attention module. %, as shown in Fig. \ref{AttnNetwSchematic}. 
The output of this module is $\mathbf{X} \in \mathbb{R}^{N\times T \times D}$, got by stacking $\mathbf{x}_i, i\in {0, \dots, N-1}$. 
The temporal attention process can be described by a function $f_{tem}(\mathbf{E})\rightarrow \mathbf{X}$. 

The output of the temporal attention module results in an assessment of %(the embedding of) 
each agent's observation at any single time-step relative to observations at other time-steps of an episode. 
To obtain further insight into how an agent's observation is related to other agents' observations, an agent-attention module is concatenated to the temporal-attention module.   
\subsubsection{\textbf{Agent-Attention Module}}
%Similar to the temporal-attention module, the 
The agent-attention module uses the transpose of $\mathbf{X}$, denoted $\bar{\mathbf{X}} \in \mathbb{R}^{T \times N \times D}$. 
Each row of $\bar{\mathbf{X}}$, $\mathbf{x}_t \in \mathbb{R}^{N \times D}$ is transformed to a \emph{query} $Q^{agt}_t=\mathbf{x}_tW^{agt}_q$, \emph{key} $K^{agt}_t=\mathbf{x}_tW^{agt}_k$, and \emph{value} $V^{agt}_t=\mathbf{x}_tW^{agt}_v$. 
Here, $W^{agt}_q, W^{agt}_k, W^{agt}_v \in \mathbb{R}^{D \times D}$ are learnable parameters. 
The $i^{th}$ row $z_{t,i}$ of the agent attention feature $\mathbf{z}_t \in \mathbb{R}^{N \times D}$ is a weighted sum, % of values:
$z_{t,i} = \beta _{t,i}^T V^{agt}_t$.
%\begin{align}
%	z_{t,i} = \beta _{t,i}^T V^{agt}_t.
%\end{align}
Maintaining causality is not necessary when computing the agent attention weight vector $\beta _{t,i} \in \mathbb{R}^{N\times 1}$. These weights are determined similar to the temporal attention weight vector in Eqn. (\ref{TempAttnVect}), except without a masking operation. Therefore,
%
%It should be noted that when computing attention weight vector $\beta_i \in \mathbb{R}^{N\times 1}$, causality is not required and therefore:
\begin{align}
	\beta _{t,i}^T = \text{softmax} (\frac{ q^{agt} _{t,i} {K^{agt}_t }^T}{\sqrt{D}}). \label{AgAttnVect}
\end{align}
The agent attention procedure can be described by a function $f_{agt}(\mathbf{X})\rightarrow \mathbf{Z}$, where $\mathbf{Z} \in \mathbb{R}^{T \times N \times D}$.
\subsubsection{\textbf{Concatenating Attention Modules}}
%We concatenate the temporal and agent attention modules in order to fuse information from the individual modules. 
The output of the temporal attention module is an entity $\mathbf{X}$ that attends to information at time-steps along the length of an episode for each agent. 
Passing $\mathbf{X}$ through the agent attention module results in an output $\mathbf{Z}$ that is attended to by embeddings at all time-steps and from all agents. 
The data-flow of this process can be written as a composition of functions: $f_{att}:=f_{agt}\circ f_{tem}$. 
%
%The mapping $f_{atac}$ is obtained by repeated composition of the temporal and agent attention modules. 
%That is, $f_{atac}:= f_{att} \circ \dots \circ f_{att}$. 
The temporal and agent attention modules can be repeatedly composed to improve expressivity. 
The position embedding is required only at the first temporal attention module when more than one is used. 
\subsection{Credit Assignment}

The output of the attention modules is used to assign credit at each time-step along the length of the episode. 
Let $f_{arel}:= f_c \circ (f_{att} \circ \dots \circ f_{att})$, where $f_c: \mathbb{R}^{T \times N \times D} \rightarrow \mathbb{R}^T$. 
In order to carry out temporal credit assignment effectively, we leverage a property of permutation invariance. 
%We discuss the credit assignment process in detail in the following. 
%
%After obtaining the attention feature $\mathbf{Z} \in \mathbb{R}^{T \times N \times D}$, we need to assign a credit to each time step.
%Let $f_c$ denote the function mapping $\mathbf{Z}$ to $\mathbf{r} \in \mathbb{R}^T$.
%Then the overall credit assignment function $f=f_c \circ f_{agt}\circ f_{tem}$.
%In the following, we first introduce permutation invariance in multi-agent reinforcement learning.
%
\subsubsection{\textbf{Permutation Invariance}}
Agents sharing the same action and observation spaces are termed \emph{homogeneous}. 
When homogeneous agents $ag1$ and $ag2$ cooperate to achieve a goal, 
%the identities of agents that visit a state should not matter. 
%For example, 
the reward when $ag1$ observes $ob1$ and $ag2$ observes $ob2$  
%the reward when homogeneous agents $ag1$ and $ag2$ have observations $ob1$ and $ob2$ 
should be the same as the case when $ag1$ observes $ob2$ and $ag2$ observes $ob1$. 
This property is called \emph{permutation invariance}, and 
%Permutation invariance 
has been shown to improve the sample-efficiency of multi-agent credit assignment as the number of agents increase \cite{liu2020pic, gangwani2020learning}.  
When this property is satisfied, %one can expect that 
the output of the function $f_{arel}$ should be invariant to the order of the agents' observations. % of the agents. 
Formally, if the set of all permutations along the agent dimension (second dimension of $\mathbf{E}$) is denoted $\mathcal{H}$, then $f_{arel}(h_1(\mathbf{E})) = f_{arel}(h_2(\mathbf{E}))$ must be true for all $h_1, h_2 \in \mathcal{H}$.
%$h(f_{arel}(\mathbf{E}))$ should be equal to $f_{arel}(h(\mathbf{E}))$, where $h \in \mathcal{H}$. 
%We define the set of all possible shuffling function along agent dimension (the second dimension of $\mathbf{Z}$) as $\mathcal{H}$, then $f(\mathbf{Z})$ should be equal to $f(h(\mathbf{Z})), h\in\mathcal{H}$. \textcolor{red}{give a formal definition of h?}

The function $f_{att}$ is permutation invariant along the agent dimension by design. 
%Therefore, in order to enable permutation invariant credit assignment, one sufficient condition is that the function $f_c$ is permutation invariant. 
A sufficient condition for $f_{arel}$ to be permutation invariant %credit assignment i
is that the function $f_c$ be permutation invariant. 
%This can be accomplished in a number of ways. 
%There are many ways to achieve this.
%For example, we can calculate the mean along the second two dimensions of $\mathbf{Z}$. 
%Or more sophistic method can be adopted, like graph convolution networks \cite{liu2020pic} on $\mathbf{Z}$ along the agent dimension.
%In this work, 
%The authors of \cite{zaheer2017deep} established necessary and sufficient conditions for permutation invariance in neural networks. 
%In our paper,  
To ensure this, 
we apply a multi-layer perceptron (MLP), %shared by each agent dimension, 
add the MLP outputs element-wise, and pass it through another MLP. 
%\textcolor{red}{The outputs are summed up element-wise and then passed through another MLP.}
%This is denoted by a function $g:\mathbb{R}^{D}\rightarrow \mathbb{R}$, %then simply sum up the output of MLP.
When functions $g_1$ and $g_2$ associated to the MLPs are continuous and shared among agents, the evaluation at time $t$ is the \emph{predicted reward} $\hat{r}_t := g_2\big(\sum_{i=0}^{N-1}g_1(z_{t,i})\big)$. 
It was shown in \cite{zaheer2017deep} that any permutation invariant function can be represented by the above equation.
%\textcolor{red}{\begin{align}
%	\hat{r}_t = g_2\big(\sum_{i=0}^{N-1}g_1(z_{t,i})\big).
%	\end{align}}
%
\begin{rem}
\emph{AREL} can be adapted to the \emph{heterogeneous} case when cooperative agents are divided into homogeneous groups. %, where agents within a group share the same observation space.
Similar to a position embedding, we can apply an \emph{agent-group embedding} such that agents within a group share an agent-group embedding. 
This will maintain permutation invariance of observations within a group, while enabling identification of agents from different groups. 
\emph{AREL} will also work in the case when the multi-agent system is fully heterogeneous. 
This is equivalent to a scenario when there is only one agent in each homogeneous group. 
Therefore, \emph{AREL} can handle agent types ranging from fully homogeneous to fully heterogeneous. 
%In this way, observation permutation inside an agent group will not change the credit assignment output, and in the meanwhile agents from different groups can be identified. 
%By utilizing agent-group embedding, partial permutation invariance can be utilized.
\end{rem}
\subsubsection{\textbf{Credit Assignment Learning}}
Given a reward $R_T$ at the end of an episode of length $T$, the goal is to learn a temporal decomposition of $R_T$ %in order 
to assess contributions of agents at each time-step along the trajectory. 
Specifically, we want to learn %a sequence 
$\{\hat{r}_t\}_{t=0}^{T}$ satisfying $ \sum_{t=0}^{T}\hat{r}_t = R_T$. 
Since $f_{arel}^{\theta}( \mathbf{E} )$ is a vector in $\mathbb{R}^T$, its $t^{th}$ entry is denoted $f_{arel}^{\theta}( \mathbf{E}_t )$ ($=\hat{r}_t$). 
The sequence $\{\hat{r}_t\}_{t=0}^{T}$ 
%This sequence of redistributed rewards 
is learned by minimizing a \emph{regression loss}, 
%$l_r(\theta) := \mathbb{E}_{\mathbf{E,R_T}} \big[\big( \sum_{t}( f_{atac}^{\theta}( \mathbf{E}_t ) ) - R_T \big)^2\big]$
$l_r(\theta) := \mathbb{E}_{\mathbf{E,R_T}} \big[\frac{1}{T}\big( \sum_{t}( f_{arel}^{\theta}( \mathbf{E}_t ) ) - R_T \big)^2\big]$, where $\theta$ are neural network parameters. % associated to the neural networks.% that make up the attention modules. 

%\textcolor{blue}{
The redistributed rewards will be provided as an input to a MARL algorithm. 
We want to discourage $\{\hat{r}_t\}_{t=0}^{T}$ from being sparse, since sparse rewards may impede learning policies \cite{devlin2011theoretical}. 
%We also want to incorporate the possibility that not all intermediate states will contribute equally to the episodic reward. 
We observe that more than one combination of $\{\hat{r}_t\}_{t=0}^{T}$ can minimize $l_r (\theta)$. 
We add a regularization loss $l_{reg}(\theta)$ to select among solutions that minimize $l_r(\theta)$. 
Specifically, we aim to choose a solution that also minimizes the variance $l_v(\theta)$ of the redistributed rewards, and set $l_{reg}(\theta) = l_v(\theta)$ (we examine other choices of $l_{reg}(\theta)$ in the \emph{Appendix}). 
Using $l_v(\theta)$ as a regularization term in the loss function leads to learning less sparsely redistributed rewards. 
%}
With $\omega \in \mathbb R_{\ge 0}$ denoting a hyperparameter, %to balance the relative significance of these two terms, 
the combined loss function used to learn $f_{arel}^{\theta}$ is:% given by: 
\begin{align}
	loss_{total}(\theta) = l_r(\theta) + \omega l_v (\theta), \label{LossTotal}
\end{align}
%
%\textcolor{blue}{
where $l_v(\theta):= \mathbb{E}_{\mathbf{E}}\big[\frac{1}{T}\sum_{t}(f_{arel}^{\theta}(\mathbf{E}_t)-\bar{f}^{\theta}(\mathbf{E}))^2 \big]$, and $\bar{f}^{\theta}(\mathbf{E}):=$\\ $ \big(\sum_t f_{arel}^{\theta}(\mathbf{E}_t)\big)/T$. 
The form of $loss_{total}(\theta)$ in Eqn. (\ref{LossTotal}) incorporates the possibility that not all intermediate states will contribute equally to $R_T$, 
%the episodic reward, 
and additionally results in learning less sparsely redistributed rewards. %, while additionally incorporating possibility that not all intermediate states would contribute equally to the episodic reward. 
%}
%
%\textcolor{blue}{
Note that $\arg \min_{\theta}[loss_{total}(\theta)]$ %minimizing Eqn. (\ref{LossTotal}) 
will not typically yield $l_r(\theta) =l_v (\theta)= 0$ %simultaneously 
(which corresponds to a uniform redistribution of rewards). 
Since some states may be common to different episodes, the redistributed reward $\hat{r}_t$ at each time-step 
cannot be arbitrarily chosen. 
%Since some states may be common to different episodes, a uniform reward redistribution may not be valid. 
For e.g., consider $N$ different episodes $\{E_i\}_{i=1}^N$, each of length $L$, with distinct cumulative episodic rewards $R_i$. 
If an intermediate state $s$ is common to episodes $E_j$ and $E_k$, %(as it is typical in experiments for more than one episode to share a common state), 
under a uniform redistribution, distinct rewards $R_j/L$ and $R_k/L$ will be assigned to $s$, which is not possible. %the same state. 
Thus, $l_r(\theta)=0$ and $l_v(\theta)=0$ will not both be true, implying that %in practice, 
a uniform redistribution may not be viable.
%}
%
\subsection{Algorithm}
\begin{algorithm} [!ht]
	\caption{\emph{AREL}}%: Agent-Temporal attention for multi-agent credit assignment}
	\begin{algorithmic}[1] \label{Algo1}
		\renewcommand{\algorithmicrequire}{\textbf{Input:}}
		\REQUIRE Number of agents $N$. Reward weight $\alpha \in [0,1]$. \\
		Initialize parameters $\theta, \phi$ for credit assignment and RL  (policy/ critic) modules respectively. \\
		%Initialize parameters $\phi$ for RL module (policy/ critic). \\ 
		Experience buffer for storing trajectories $B_e\leftarrow \emptyset$.  Prediction function update frequency $M$.
		\FOR{Episode $k=0,\dots$}
		\STATE Reset episode return $R_T\leftarrow0$; Reset trajectory for current episode $\tau\leftarrow\emptyset$
%		\STATE Reset trajectory for current episode $\tau\leftarrow\emptyset$
		\FOR{step $t=0,\dots,T-1$}
		\STATE Sample action $a^i_t\sim \pi^i_{\phi}(o^i_t)$, for $i=0,\dots,N-1$
		\STATE Take action $a_t$; Observe $o_{t+1}=(o^0_{t+1},\dots,o^{N-1}_{t+1})$
		\STATE Store transition $\tau\leftarrow\tau \cup \{(o_t, a_t,o_{t+1})\}$
		\ENDFOR
		\STATE Update episode reward $R_T$; Store trajectory $B_e \leftarrow B_e \cup {(\tau, R_T)}$
%		\STATE Store trajectory $B_e \leftarrow B_e \cup {(\tau, R_T)}$
%		\IF{using on-policy RL algorithm}
%		\STATE Predict reward $\hat{r}_t$ using $f_{\theta}(\tau)$
%		\STATE Update $\phi$ using $\{(o_t,a_t,o_{o_{t+1}})\} \in \tau$ and corresponding $\hat{r}_t$
%		\ENDIF
%		\IF{using off-policy RL algorithm}
		\STATE Sample a batch of trajectories $B_{\tau}\sim B_e$
		\STATE Predict reward $\hat{r}_t$ using $f^{\theta}_{arel}(\tau)$ for each $\tau \in B_{\tau}$
		\STATE Update $\phi$ using $\{(o_t,a_t,o_{o_{t+1}})\} \in \tau$ and weighted reward $\alpha \hat{r}_t + (1-\alpha)\mathbf{1}_{t=T}R_T$
%	    \ENDIF
		\IF{$k$ mod $M$ is 0}
		\FOR{each gradient update}
		\STATE Sample a batch from $B_e$, and compute estimate of total loss, $\hat{loss}_{total}(\theta)$
		%\STATE Update $\theta$ by minimizing the total loss equation  (\ref{LossTotal})
		\STATE $\theta \leftarrow \theta - \nabla_{\theta} \hat{loss}_{total}(\theta)$
		\ENDFOR
		\ENDIF
		\ENDFOR
	\end{algorithmic}
\end{algorithm}
\emph{AREL} is summarized in Algorithm \ref{Algo1}. 
Parameters $\phi$ of RL modules %(e.g., policy, critic parameters) 
and $\theta$ of the credit assignment function are randomly initialized.
Observations and actions of agents are collected in each episode (Lines 2-6). 
Trajectories and episodic rewards are stored in an experience buffer $B_e$ (Line 8). 
The reward $\hat{r}_t$ at each time step for every trajectory in a batch $B_{\tau}$ (sampled from $B_e$) is predicted (Lines 9-10). 
The predicted $\hat{r}_t$ changes as $\theta$ is updated, but the episode reward $R_T$ remains the same. 
A weighted sum $\alpha \hat{r}_t + (1-\alpha)\mathbf{1}_{t=T}R_T$ ($\mathbf{1}_{t=T}$ is an indicator function) is used to update %the parameter 
$\phi$ in a stable manner by using a MARL algorithm (Line 11). 
%, and 
%The update of $\phi$ 
%can be carried out using any MARL algorithm. 
The credit assignment function $f_{arel}^{\theta}$ is updated when $M$ new trajectories are available (Lines 12-17).

\subsection{Analysis}\label{Sec:Analysis}
\begin{figure*}
	\begin{subfigure}{0.31\textwidth}
		\includegraphics[width=\linewidth]{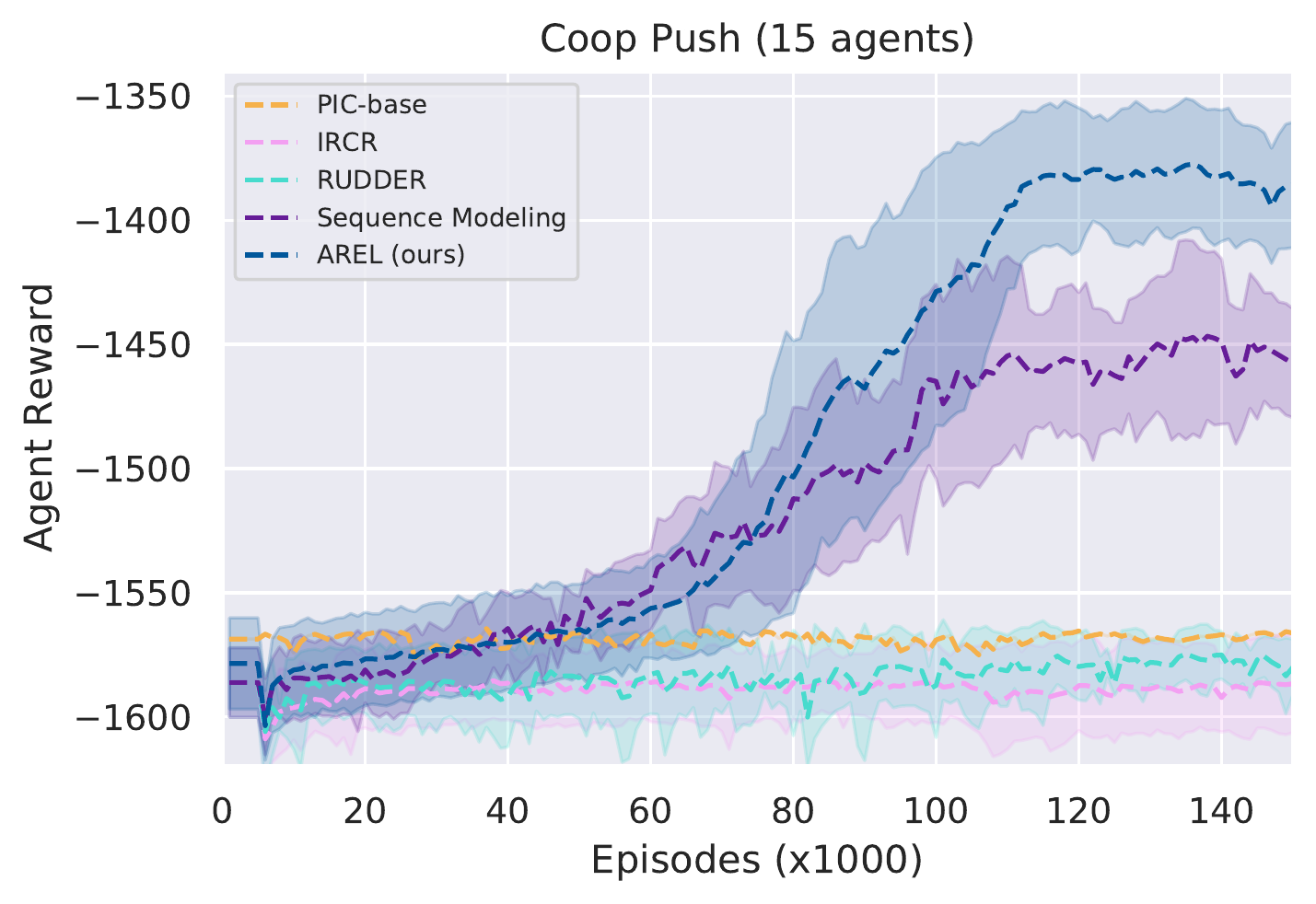}
		\caption{Cooperative Push} \label{fig:2a}
	\end{subfigure}%
	\hspace*{\fill}   % maximize separation between the subfigures
	\begin{subfigure}{0.31\textwidth}
		\includegraphics[width=\linewidth]{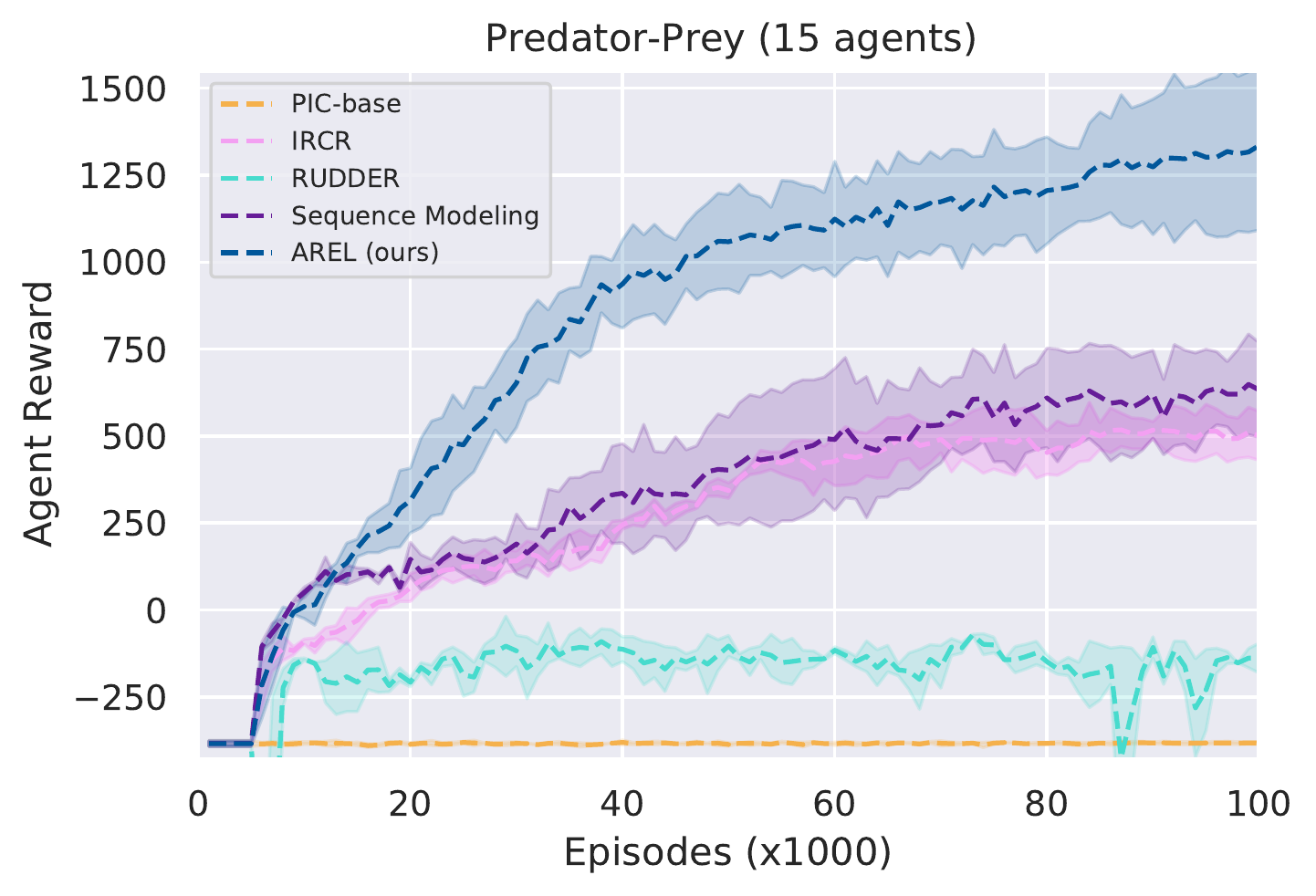}
		\caption{Predator and Prey} \label{fig:2b}
	\end{subfigure}%
	\hspace*{\fill}   % maximizeseparation between the subfigures
	\begin{subfigure}{0.31\textwidth}
		\includegraphics[width=\linewidth]{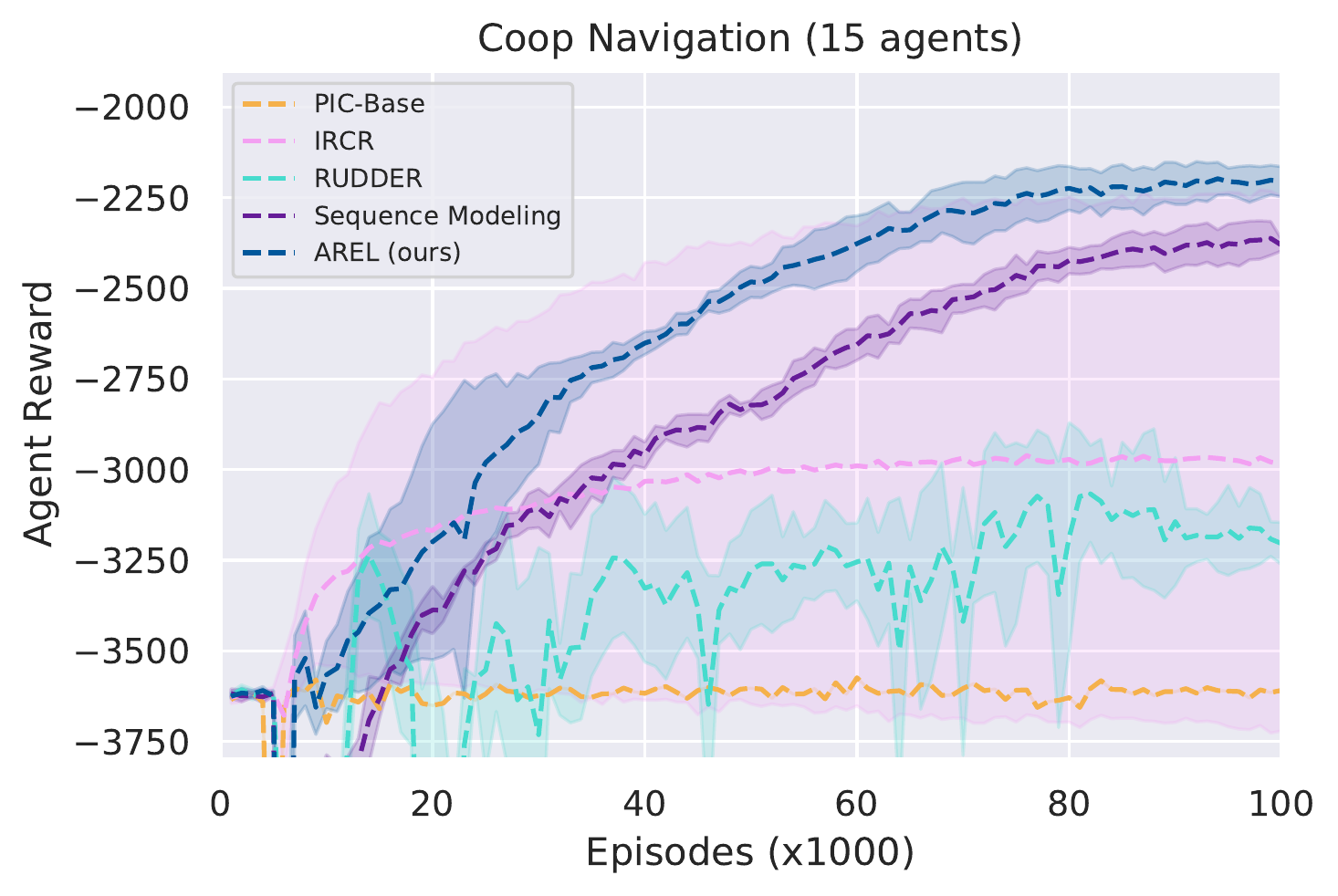}
		\caption{Cooperative Navigation} \label{fig:2c}
	\end{subfigure}
	\caption{Average agent rewards and standard deviation for tasks in Particle World with episodic rewards and $N=15$. \emph{AREL} (dark blue) results in the highest average rewards in all tasks. %In comparison, the PIC baseline (orange) and RUDDER (green) are unable to learn policies to complete the tasks. Using a surrogate objective in Trajectory Smoothing (pink) results in comparable rewards in the Cooperative Navigation task, but the reward curve has a high variance. Sequence modeling (purple) obtains a higher reward than the former three methods, but does not explicitly model agent-attention, which results in a lower average reward than \emph{AREL}.
} \label{FigGraphsN15}
\end{figure*}
In order to establish a connection between redistributed rewards from \emph{Line 10} of Algorithm \ref{Algo1} and the episodic reward, we define \emph{return equivalence} of \emph{decentralized partially observable sequence-Markov decision processes (Dec-POSDP)}. 
This generalizes the notion of return equivalence introduced in \cite{arjona2019rudder} in the fully observable setting for the single agent case. 
%
%In this section, we define a notion of \emph{return equivalence} \cite{arjona2019rudder} on \emph{decentralized partially observable sequence-Markov decision process (Dec-POSDP)}. 
%
A Dec-POSDP is a decision process with Markov transition probabilities but has a reward distribution that need not be %required to be 
Markov.
We present a result that establishes that return equivalent Dec-POSDPs will have the same optimal policies. %, and 
%We also 
%provide a bound on the loss in Eqn. (\ref{LossTotal}) when estimates of the predicted rewards are unbiased. 
%The proof and detailed analyses are provided in \emph{Appendix A}.
%
%In this section, we introduce \emph{sequence decision processes}, an entity that has Markovian transition probabilities, but does not require rewards to be Markov. 
%Sequence decision processes with different reward structures are said to be \emph{return-equivalent} if the expected returns for any trajectory on these processes are equal. 
%We present a result that establishes that return equivalent sequence decision processes will have the same optimal policies. 
%We also provide an interpretation of the overall loss in Eqn. (\ref{LossTotal}) in terms of a bias-variance trade-off in a mean square estimator. 
%
%\subsection{Return-equivalent Stochastic Game}
%\subsection{Return-equivalent Partially observable Sequence-Markov Decision Process}
%Similar to the definition of sequence-Markov decision process (SDP) in \cite{arjona2019rudder}, a decentralized partially observable sequence-Markov decision process (Dec-POSDP) is a decision process that has Markov transition probabilities but a reward probability that is not required to be Markov.
%We can define return-equivalent Dec-POSDPs which have the same optimal policies.
\begin{df}\label{DecPOSDPDefn}
%	Two partially observable sequence-Markov decision processes $\tilde{\mathcal{P}}$ and $\mathcal{P}$ are return-equivalent if they differ only in their reward but for each joint policy $\pi$ have the same expected return $\tilde{V}^{\pi}(s_0)=V^{\pi}(s_0)$ at $t=0$.
%	$\tilde{\mathcal{P}}$ and $\mathcal{P}$ are strictly return-equivalent if they have the same expected return for any trajectory:
Dec-POSDPs $~\tilde{\mathcal{P}}$ and $\mathcal{P}$ are \textbf{return-equivalent} if they differ only in their reward functions but have the same return for any trajectory $\tau$.% That is, $\tilde{R}_0(\tau) = R_0(\tau), \forall \tau$. 
%	\begin{align}
%		\tilde{R}_0(\tau) = R_0(\tau), \forall \tau \nonumber.
%	\end{align}
%	
\end{df}
\begin{thm}\label{RetEqThm}
Given an initial state $s_0$, return-equivalent Dec-POSDPs will have the same optimal policies. 
%	Two partially observable sequence-Markov decision processes $\tilde{\mathcal{P}}$ and $\mathcal{P}$ that are return-equivalent will have the same optimal policies.
\end{thm}
According to Definition \ref{DecPOSDPDefn}, any two return equivalent Dec-POSDPs will have the same expected return for any trajectory $\tau$. 
That is, $\tilde{R}_0(\tau) = R_0(\tau), \forall \tau$. 
This is used to prove Theorem \ref{RetEqThm}. 

%\begin{proof}[\textbf{Theorem \ref{RetEqThm}}]
\begin{proof}
	Consider two return-equivalent Dec-POSDPs $\tilde{\mathcal{P}}$ and $\mathcal{P}$. 
	Since $\tilde{\mathcal{P}}$ and $\mathcal{P}$ 
	have the same transition probability and observation functions, the probabilities that a trajectory $\tau$ is realized will be the same if both Dec-POSDPs are provided with the same policy. 
	For any joint agent policy $\pi:=(\pi^1,\dots,\pi^n)$ and sequence of states $s:=(s_0, \dots, s_{T})$, we have:
%	The joint agent policy and a sequence of states are denoted as $\pi:=(\pi^1,\dots,\pi^n)$ and $s:=(s_0, \dots, s_{T})$, respectively.
%	For any given common policy $\pi$, we have:
	\begin{align}
	\mathbb{E}_{\tau \sim (\pi,\tilde{Z},\tilde{P})}&\big[\tilde{R}_0(\tau)\big] \nonumber \\
	&=\sum_{\tau}\tilde{R}_0(\tau) \underbrace{\sum_{s}\tilde{p}(s_0)\prod_{t=0}^{T-1}\pi(a_t|o_t)\tilde{Z}(o_t|s_t)\tilde{p}(s_{t+1}|s_t,a_t)}_{\tilde{p}^{\pi}(\tau)} \nonumber\\		
	&= \sum_{\tau}\tilde{R}_0(\tau) \underbrace{\sum_{s}p(s_0)\prod_{t=0}^{T-1}\pi(a_t|o_t)Z(o_t|s_t)p(s_{t+1}|s_t,a_t)}_{p^{\pi}(\tau)} \nonumber\\
	&= \sum_{\tau}R_0(\tau) \underbrace{\sum_{s}p(s_0)\prod_{t=0}^{T-1}\pi(a_t|o_t)Z(o_t|s_t)p(s_{t+1}|s_t,a_t)}_{p^{\pi}(\tau)} \nonumber\\
	&=\mathbb{E}_{\tau \sim (\pi,Z,P)} \big[R_0(\tau)\big]. \nonumber
	\end{align}
	These equations follow from %the definition of return-equivalent Dec-POSDPs in 
	Definition \ref{DecPOSDPDefn}.
	Let $\pi^*$ denote an optimal policy for $\tilde{\mathcal{P}}$. Then, we have:
	\begin{align}
	\mathbb{E}_{\tau \sim (\pi^*,\tilde{Z},\tilde{P})}\big[\tilde{R}_0(\tau)\big]&=\mathbb{E}_{\tau \sim (\pi^*,Z,P)}\big[R_0(\tau)\big] \nonumber\\
	\geq \mathbb{E}_{\tau \sim (\pi,\tilde{Z},\tilde{P})}\big[\tilde{R}_0(\tau)\big]&=\mathbb{E}_{\tau \sim (\pi,Z,P)}\big[R_0(\tau)\big]. \nonumber
	\end{align}
	Therefore, $\pi^*$ will also be an optimal policy for $\mathcal{P}$.%, which completes the proof.
\end{proof} 

When $l_r(\theta) = 0$ in Eqn. (\ref{LossTotal}), a Dec-POSDP with the redistributed reward will be return-equivalent to a Dec-POSDP with the original episodic reward. 
Theorem \ref{RetEqThm} indicates that in this scenario, the two Dec-POSDPs will have the same optimal policies. 
An additional result in \emph{Appendix A} 
%The following result 
gives a bound on $loss_{total}(\theta)$ when the estimators $f^\theta_{arel}(E_t)$ are unbiased at each time-step. 

\section{Experiments}
%\begin{figure*}
%	\begin{subfigure}{0.31\textwidth}
%		\includegraphics[width=\linewidth]{figures/coop_push_N15.pdf}
%		\caption{Cooperative Push} \label{fig:2a}
%	\end{subfigure}%
%	\hspace*{\fill}   % maximize separation between the subfigures
%	\begin{subfigure}{0.31\textwidth}
%		\includegraphics[width=\linewidth]{figures/predator_prey_N15.pdf}
%		\caption{Predator and Prey} \label{fig:2b}
%	\end{subfigure}%
%	\hspace*{\fill}   % maximizeseparation between the subfigures
%	\begin{subfigure}{0.31\textwidth}
%		\includegraphics[width=\linewidth]{figures/simple_spread_N15.pdf}
%		\caption{Cooperative Navigation} \label{fig:2c}
%	\end{subfigure}
%%	
%	\caption{Average agent rewards and standard deviation for tasks in Particle World with episodic rewards and $N=15$. \emph{AREL} (dark blue) results in the highest average rewards in all tasks. %In comparison, the PIC baseline (orange) and RUDDER (green) are unable to learn policies to complete the tasks. Using a surrogate objective in Trajectory Smoothing (pink) results in comparable rewards in the Cooperative Navigation task, but the reward curve has a high variance. Sequence modeling (purple) obtains a higher reward than the former three methods, but does not explicitly model agent-attention, which results in a lower average reward than \emph{AREL}.
%} \label{FigGraphsN15}
%\end{figure*}
%
In this section, we describe the tasks that we evaluate \emph{AREL} on, and present results of our experiments. 
Our code is available at \textbf{\url{https://github.com/baicenxiao/AREL}}.
%In this section, we describe the tasks that we evaluate \emph{AREL} on. % and details of our evaluation methods. 
%We compare %the performance of 
%\emph{AREL} with three state-of-the-art reward redistribution methods, and carry out multiple ablations. 
%We also carry out ablations to examine the impact of the different components in \emph{AREL}. 
%Additional details are presented in the \emph{Appendix} included in the \textbf{Supplementary Material}.
%
\subsection{Environments and Tasks}
%
%
%We examine three tasks from \emph{Particle World} \cite{lowe2017multi} where multiple agents share a two-dimensional space with continuous states and actions. 
%We also evaluate three maps from the \emph{StarCraft Multi-Agent Challenge} \cite{samvelyan19smac}, where agents have continuous observations and discrete actions. 
We study tasks from \emph{Particle World} \cite{lowe2017multi} and the \emph{StarCraft Multi-Agent Challenge} \cite{samvelyan19smac}. 
%As noted in \cite{lowe2017multi, samvelyan19smac}, these  
These have been identified as challenging multi-agent environments in \cite{lowe2017multi, samvelyan19smac}. %, and encompass continuous and discrete actions. 
%These 
%are challenging multi-agent environments with partially observable states, and encompass continuous and discrete actions, as noted in \cite{lowe2017multi, samvelyan19smac}.  
In each task, a reward is received by agents only at the end of an episode. 
No reward is provided at other time steps. 
%The delayed %nature of the 
%reward makes it difficult for agents to receive immediate feedback on the quality of their actions at each time step. 
We briefly summarize the tasks below and defer detailed task descriptions to \emph{Appendix B}.\\
(1) \textbf{Cooperative Push}: $N$ agents work together to move a large ball to a landmark. \\
(2) \textbf{Predator-Prey}: $N$ predators seek to capture $M$ preys. $L$ landmarks impede movement of agents.\\ 
%Predators are rewarded when they collide with a prey. \\
(3) \textbf{Cooperative Navigation}: $N$ agents seek to reach $N$ landmarks. The maximum reward is obtained when there is exactly one agent at each landmark.\\
(4) \textbf{StarCraft}: Units from one group (controlled by RL agents) collaborate to attack units from another (controlled by heuristics).
We report results for three maps: 2 Stalkers, 3 Zealots (2s3z); 1 Colossus, 3 Stalkers, 5 Zealots (1c3s5z); 3 Stalkers vs. 5 Zealots (3s\_vs\_5z).
\subsection{Architecture and Training}
In order to make the agent-temporal attention module more expressive, we use a transformer architecture with multi-head attention \cite{vaswani2017attention} for both agent and temporal attention.
%Specifically, %in our experiments, 
%The transformer architecture applies, in sequence: an attention layer, layer normalization, two feed forward layers with ReLU activation, and another layer normalization.
%Before each layer normalization, residual connections are added. 
The permutation invariant critic (\emph{PIC}) based on the multi-agent deep deterministic policy gradient (MADDPG) from \cite{liu2020pic} is used as the base RL algorithm in Particle World. 
In StarCraft, we use QMIX \cite{rashid2018qmix} as the base RL algorithm. 
The value of $\alpha$ is set to $1$ in Particle World and $0.8$ in StarCraft. 
Additional details are presented in \emph{Appendix C}.
%Following MADDPG \cite{lowe2017multi}, the actor policy is parameterized by a two-layer MLP with 128 hidden units per layer, and ReLU activation function. 
%The permutation invariant critic is a two-layer graph convolution net with 128 hidden units per layer and a max pooling at the top, and ReLU activation. 
%The learning rates for the actor and critic are 0.01. 
%The learning rate is linearly decreased to zero at the end of training. 
%
\subsection{Evaluation}
%
%We use the permutation invariant critic (\emph{PIC}) from \cite{liu2020pic} as the baseline, and 
%We compare the performance of  \emph{AREL} against the \emph{PIC} baseline  from \cite{liu2020pic} (i.e., in the absence of an attention mechanism). 
We compare \emph{AREL} with three state-of-the-art methods:\\%, summarized below:\\ 
(1) \textbf{\emph{RUDDER}} \cite{arjona2019rudder}: A long short-term memory (LSTM) network is used for reward decomposition along the length of an episode.\\% through `contribution analysis'.\\
(2) \textbf{\emph{Sequence Modeling}} \cite{liu2019sequence}: %A transformer-based 
An 
attention mechanism is used for temporal decomposition of rewards along an episode.\\
(3) \textbf{\emph{Iterative Relative Credit Refinement (IRCR)}} \cite{gangwani2020learning}: `Guidance rewards' for temporal credit assignment are learned using a surrogate objective.
%A surrogate objective is used to learn `guidance rewards' for temporal credit assignment.
%\begin{itemize}
%\item \emph{RUDDER} \cite{arjona2019rudder}: This method uses a long short-term memory (LSTM) network for reward decomposition along the length of an episode. 
%\item \emph{Sequence Modeling} \cite{liu2019sequence}: This method uses a Transformer-based attention mechanism for temporal decomposition of rewards along an episode. 
%\item \emph{Trajectory-space Smoothing} \cite{gangwani2020learning}: This method uses a surrogate objective to learn dense `guidance rewards' for temporal credit assignment. 
%\end{itemize}

RUDDER and Sequence Modeling were originally developed for the single agent case. 
We adapted these methods to MARL by concatenating observations from all agents. 
We added the variance-based regularization loss 
%used the regularized loss from Eqn. (\ref{LossTotal}) 
in our experiments for Sequence Modeling, and observed that 
incorporating the regularization term resulted in an improved performance compared to without regularization.
% as the input.}
%We also carry out multiple ablation studies. 
%As an ablation, we evaluate the effect of removing the agent attention block, and uniformly weighting the attention of individual agents. 
%We also examine the effect of the value of the coefficient $\omega$ of the regularization term in $loss_{total}(\theta)$ (Eqn. (\ref{LossTotal})) and the number of agent-temporal attention blocks on agent rewards. 
%
%
\subsection{Results}\label{SecResults}
\begin{figure*}
	\begin{subfigure}{0.31\textwidth}
		\includegraphics[width=\linewidth]{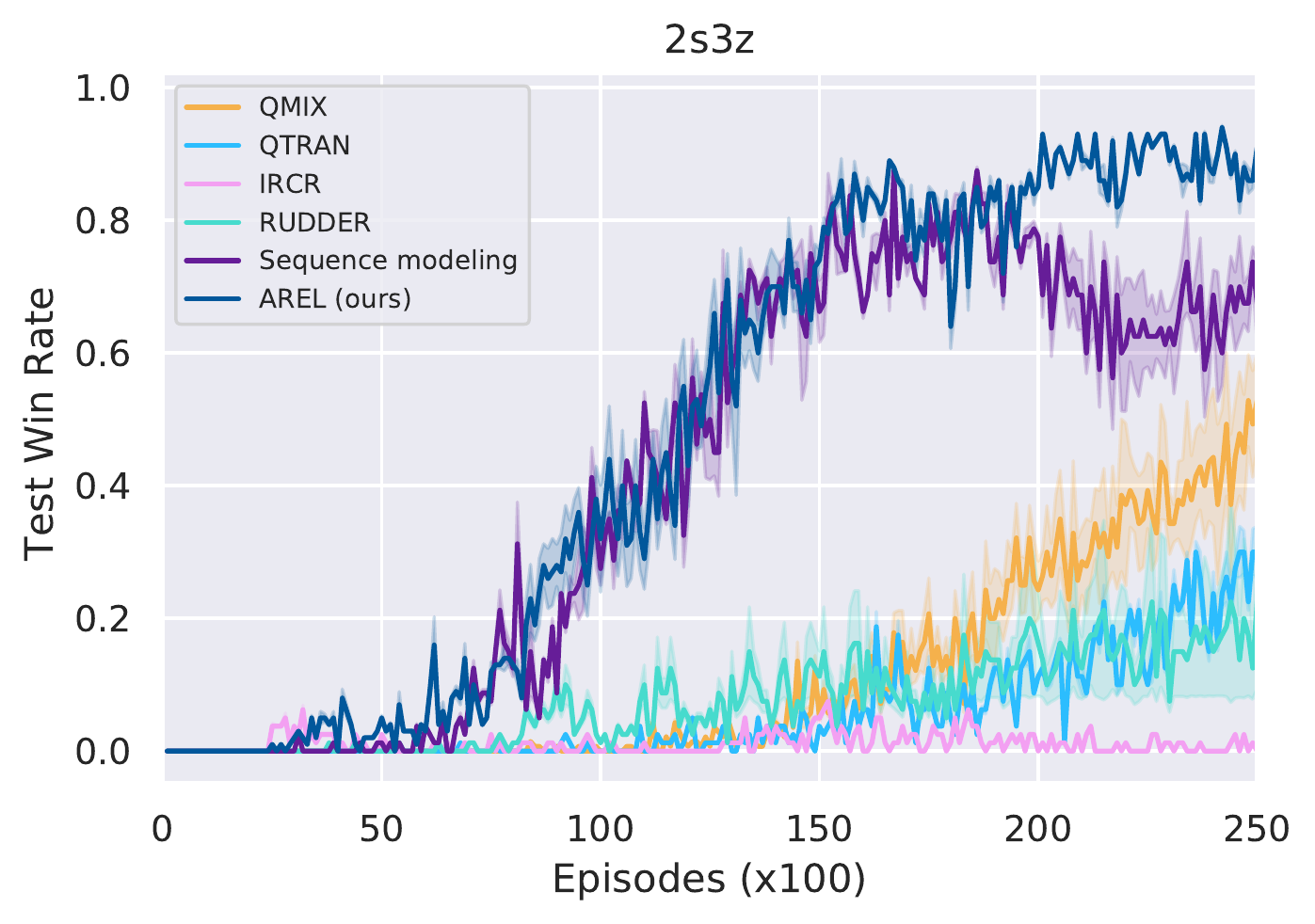}
		\caption{2s3z} \label{fig:2s3z}
	\end{subfigure}%
	\hspace*{\fill}   % maximize separation between the subfigures
	\begin{subfigure}{0.31\textwidth}
		\includegraphics[width=\linewidth]{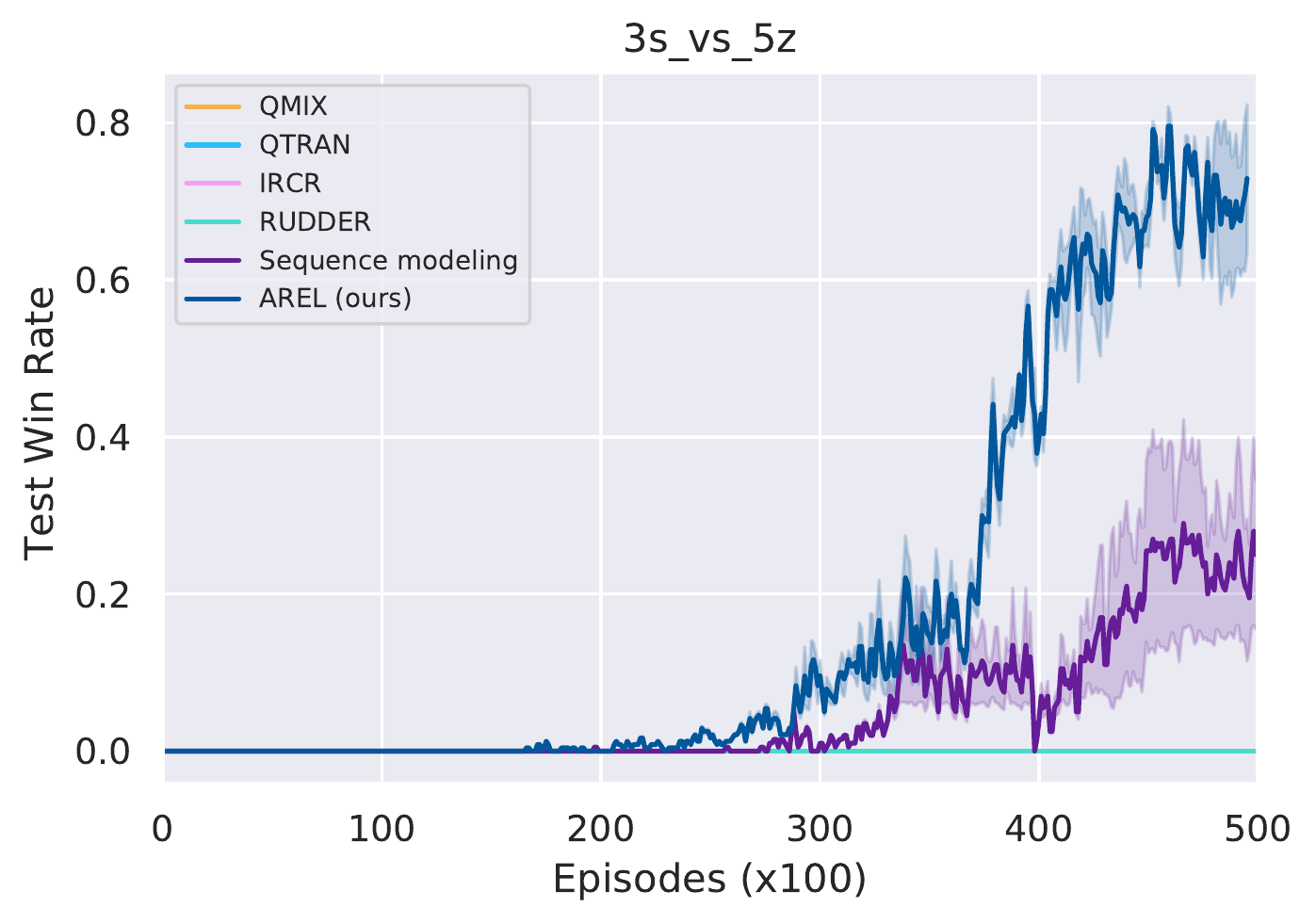}
		\caption{3s vs 5z} \label{fig:3svs5z}
	\end{subfigure}%
	\hspace*{\fill}   % maximizeseparation between the subfigures
	\begin{subfigure}{0.31\textwidth}
		\includegraphics[width=\linewidth]{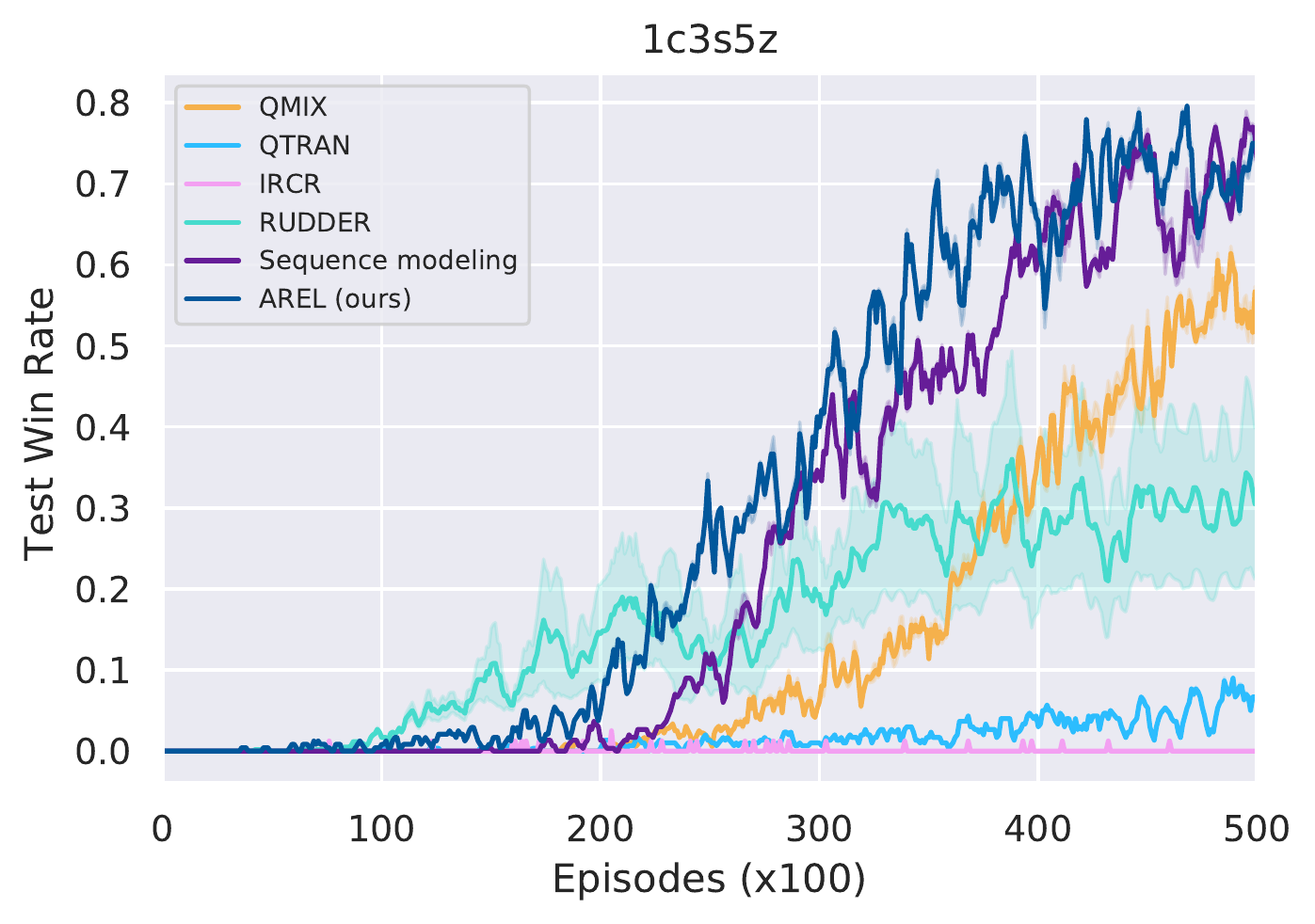}
		\caption{1c3s5z} \label{fig:1c3s5z}
	\end{subfigure}
	\caption{Average test win rate and variance in StarCraft. \emph{AREL} (dark blue) results in the highest win rates in 2s3z and 3s\_vs\_5z, and obtains a comparable win rate to Sequence Modeling in 1c3s5z. 
	} \label{FigGraphsStar}
\end{figure*}
\begin{figure*}
	\begin{subfigure}{0.31\textwidth}
		\includegraphics[width=\linewidth]{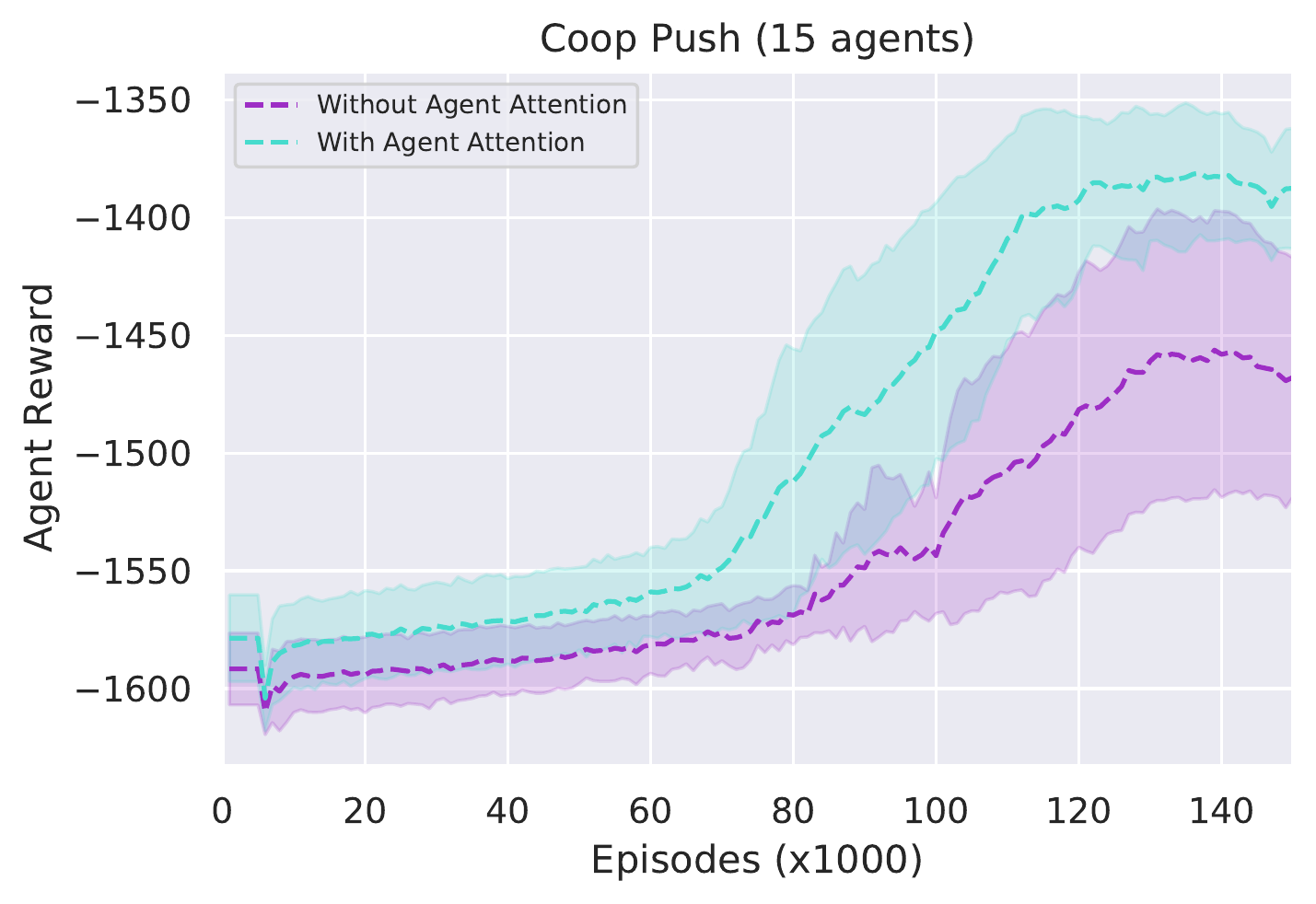}
		\caption{Agent Attention} \label{fig:Abla}
	\end{subfigure}%
	\hspace*{\fill}   % maximize separation between the subfigures
	\begin{subfigure}{0.31\textwidth}
		\includegraphics[width=\linewidth]{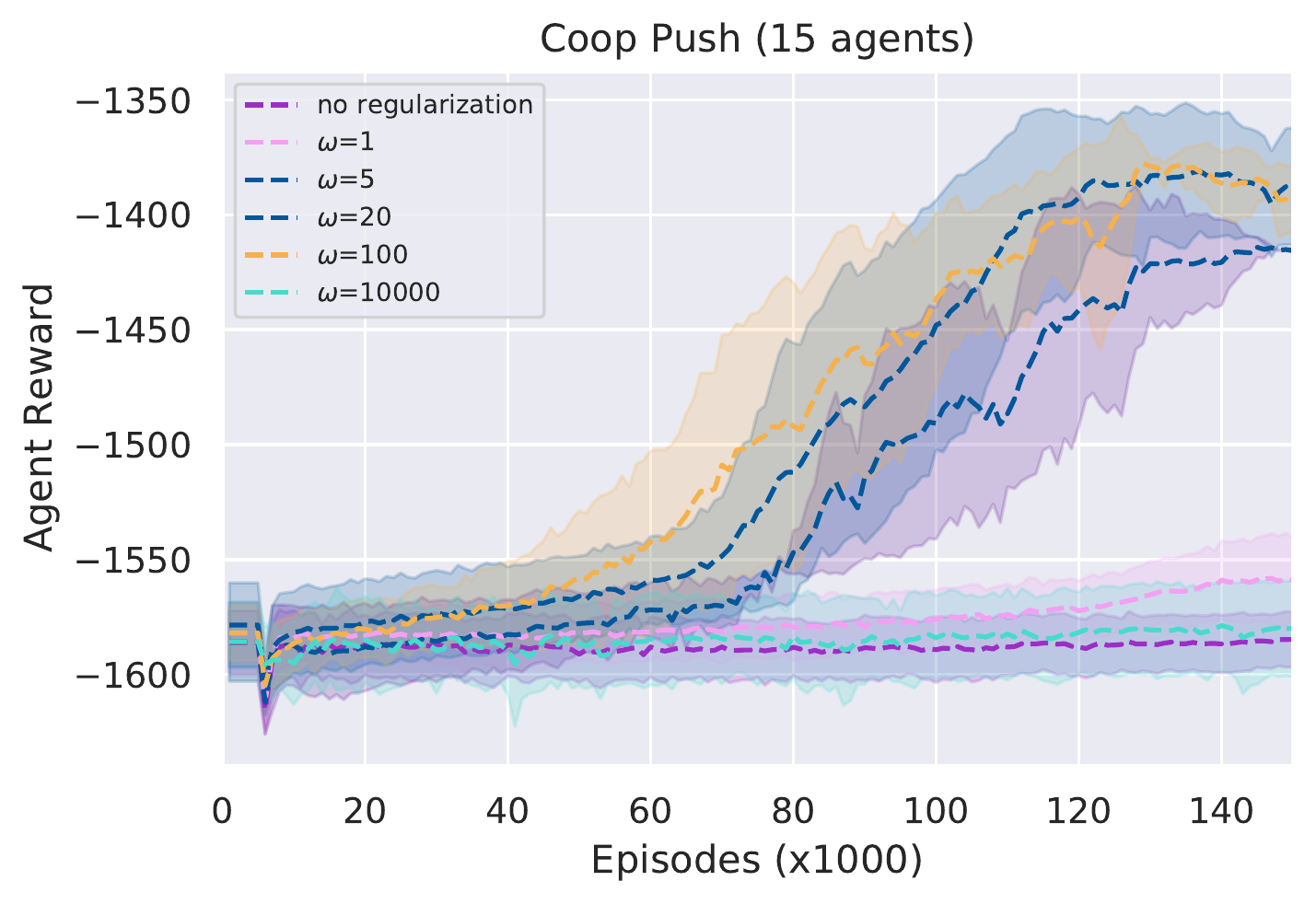}
		\caption{Regularization} \label{fig:Ablb}
	\end{subfigure}%
	\hspace*{\fill}   % maximizeseparation between the subfigures
	\begin{subfigure}{0.31\textwidth}
		\includegraphics[width=\linewidth]{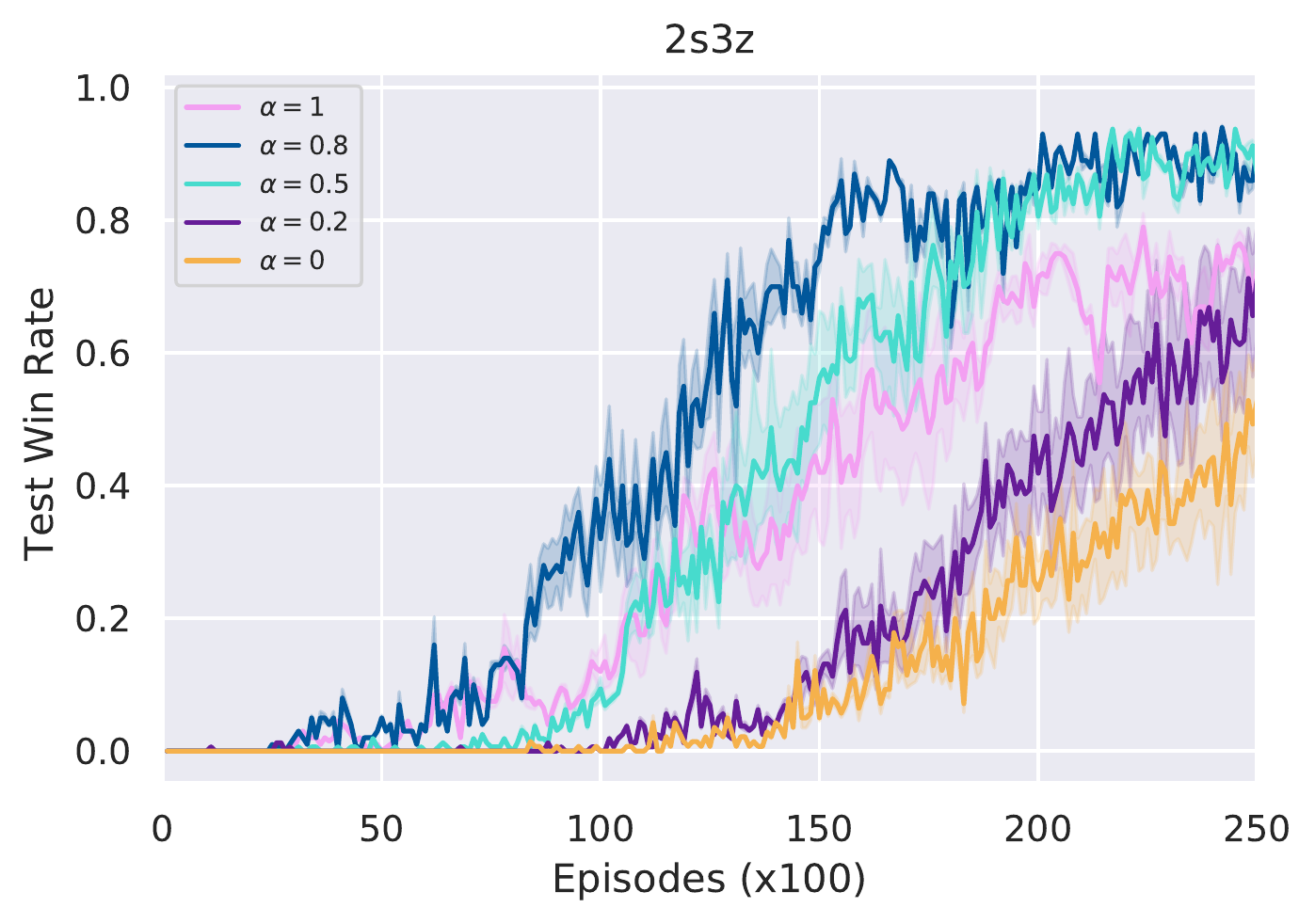}
		\caption{Reward Weight} \label{fig:Ablc}
	\end{subfigure}
	\caption{Ablations: Effects of the agent attention module (Fig. \ref{fig:Abla}) and regularization parameter $\omega$ in Eqn (\ref{LossTotal}) (Fig. \ref{fig:Ablb}) in Cooperative Push, and reward weight $\alpha$ (Fig. \ref{fig:Ablc}) in the 2s3z StarCraft map.
	} \label{FigAblations}
\end{figure*}
\subsubsection{\textbf{\emph{AREL} enables improved performance}}
Figure \ref{FigGraphsN15} shows results of our experiments for tasks in Particle World for $N = 15$. 
In each case, \emph{AREL} is able to guide agents to learn policies that result in higher average rewards 
%, and higher rewards during the early stages of training, 
compared to other methods. 
%\emph{AREL} also results in higher rewards during the early stages of training.  
This is a consequence of using an attention mechanism to redistribute an episodic reward along the length of an episode, and also characterizing the contributions of individual agents.% to the reward.

%In comparison, t
The \emph{PIC} baseline \cite{liu2020pic} fails to learn policies to complete tasks with episodic rewards. 
A similar result of failure to complete tasks was observed when using \emph{RUDDER} \cite{arjona2019rudder}. 
%when rewards are provided only at the end of an episode. 
%We observe that 
%\emph{RUDDER} \cite{arjona2019rudder} is also unable to learn policies to complete these tasks. 
An explanation for this could be that RUDDER only carries out a temporal redistribution of rewards, but does not consider the effect of  agents contributing differently to a reward. 

\emph{Sequence Modeling} \cite{liu2019sequence} performs better than \emph{RUDDER} and the \emph{PIC} baseline, possibly because it uses %an %transformer-based 
attention to redistribute episodic rewards. 
This was shown to outperform LSTM-based models, including \emph{RUDDER}, in \cite{liu2019sequence} in single-agent episodic RL, due to the relative ease of training the attention mechanism. 
%For a fair comparison, we used the regularized loss from Eqn. (\ref{LossTotal}) in our experiments. 
%Sequence modeling is seen to perform worse when the regularized term is not incorporated in the loss function. 
We believe that absence of an explicit characterization of agent-attention %is a possible justification for the 
resulted in a lower reward for this method compared to \emph{AREL}. 

Using a surrogate objective in \emph{IRCR} \cite{gangwani2020learning} results in obtaining rewards comparable to \emph{AREL} in some runs in the Cooperative Navigation task. However, the reward when using \emph{IRCR} has a much higher variance compared to that obtained when using \emph{AREL}.  %curves %are unstable and 
%has a higher variance. 
%\emph{Smooth} also obtains an average reward comparable to \emph{AREL} in the Predator-Prey task with $N = 15$. 
%We observe that 
%Our method consistently obtains higher average rewards than \emph{IRCR} during the early stages of training. 
%The performance of \emph{IRCR} is comparable to \emph{RUDDER} and the \emph{PIC} baseline in other tasks. 
A possible reason for this is that \emph{IRCR} does not characterize the relative contributions of agents at intermediate time-steps. 

Figure \ref{FigGraphsStar} shows the results of our experiments for the three maps in StarCraft. 
\emph{AREL} achieves the highest average win rate in the 2s3z and 3s\_vs\_5z maps, and obtains a comparable win rate to \emph{Sequence Modeling} in 1c3s5z. 
\emph{Sequence Modeling} does not explicitly model agent-attention, which could explain the lower average win rates %than \emph{AREL} 
in 2s3z and 3s\_vs\_5z.
\emph{RUDDER} achieves a nonzero, albeit much lower win rate than \emph{AREL} in two maps, possibly because the increased episode length might affect the redistribution of the episode reward for this method. 
%The performance of \emph{QMIX} is comparable to \emph{AREL} in two maps, but it takes a longer time to achieve the same win rate. 
%Without redistributing delayed reward, QMIX(orange) is able to achieve nonzero win rate in 2s3z and 1c3s5z but it is slower than \emph{AREL}. And in 3s\_vs\_5z, without reward redistribution, QMIX is unable to obtain nonezero win rate.
\emph{IRCR} and \emph{QTRAN} \cite{son2019qtran} obtain the lowest win rates. % in all maps. 
Additional experimental results are provided in \emph{Appendix D}.
\subsubsection{\textbf{Ablations}}
We carry out several ablations to evaluate %impact of the 
components of \emph{AREL}. % (Figure (\ref{FigAblations})).
Figure \ref{fig:Abla} demonstrates the impact of the agent-attention module. %on rewards. % in Cooperative Push for $N=15$. 
%This module is seen to play a crucial role in agents obtaining higher rewards. 
In the absence of agent-attention (while retaining permutation invariance among agents through the shared temporal attention module), rewards are significantly lower. %, both in the early stages, and at the end of training. 
We study the effect of the value of $\omega$ in %$loss_{total}(\theta)$ in 
Eqn. (\ref{LossTotal}) on rewards %in Cooperative Push 
%with $N=15$ 
in Figure \ref{fig:Ablb}. 
%The rewards are robust over a large range of values of $\omega$. 
This term is critical in ensuring that agents learn good policies. 
This is underscored by observations that rewards are significantly lower for very small or very large $\omega$ ($\omega = 0, \omega = 1, \omega = 10000$). 
%We also examine the effect of the number of agent-temporal attention blocks, $n$ (\textbf{depth}) on rewards in this task in Figure \ref{fig:Ablb}. 
%The depth has negligible impact on average rewards at the end of training. 
%However, we observe that rewards during the early stages of training are lower for $n=1$, and these rewards also have a larger variance than the other cases ($n=3, 6$). %The average rewards and variance of these rewards are similar for the latter two values of the depth.  
Third, we evaluate the effect of mixing the original episodic reward and redistributed reward by changing the reward weight $\alpha$ 
%in the 2s3z StarCraft map 
in Figure \ref{fig:Ablc}.  
The reward mixture influences win rates; $\alpha=$ 0.5 or 0.8 yields the highest win rate. The win rate is $\sim 10\%$ lower when using redistributed reward alone ($\alpha = 1$). 
Additional ablations and evaluating the %effect of the 
choice of regularization loss %in $loss_{total}(\theta)$ 
are shown in \emph{Appendices E and F}.% and \emph{Appendix F}.% respectively.

\begin{figure}[ht]
	\centering
	\includegraphics[width=0.75\linewidth]{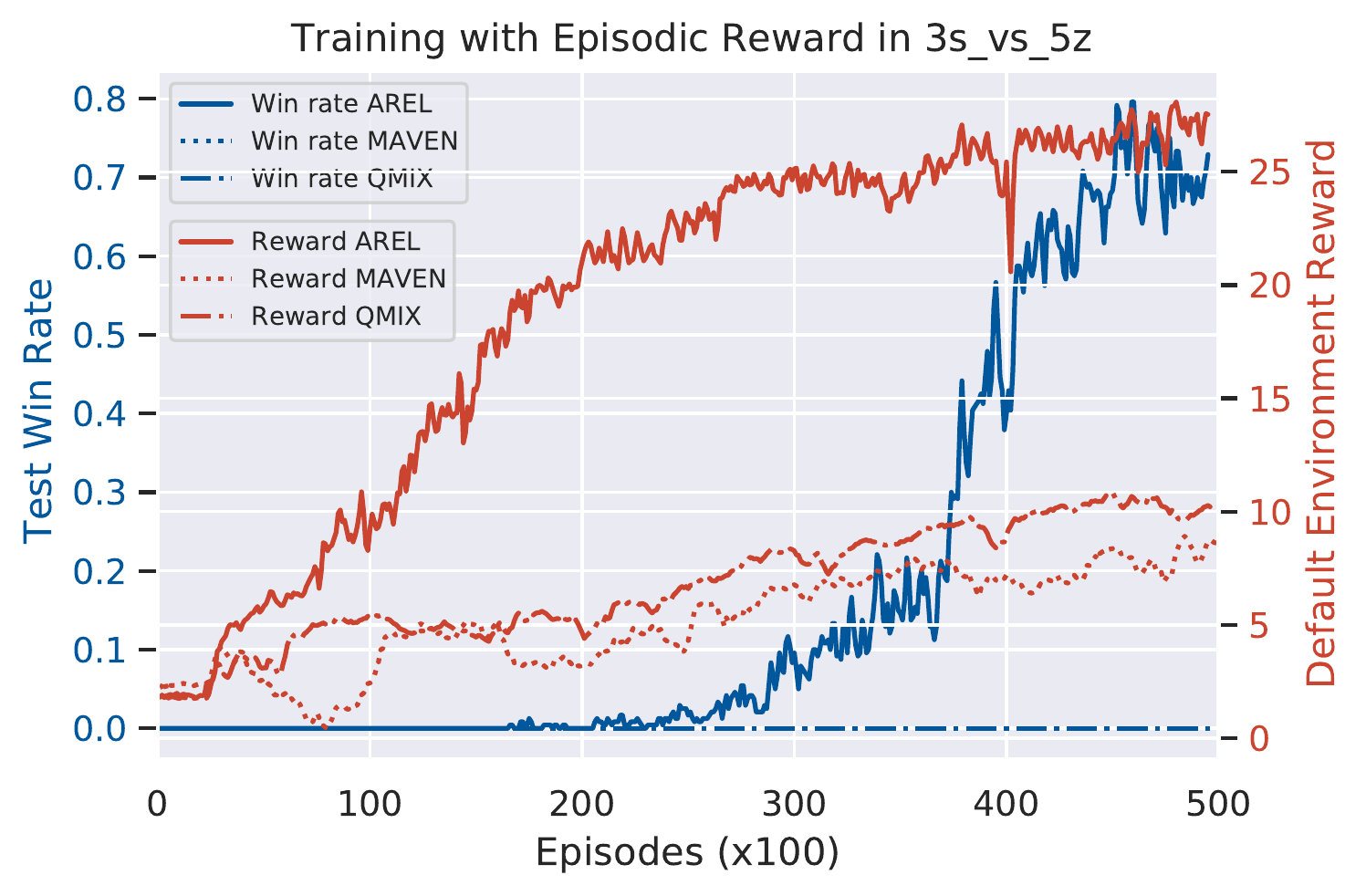}
	\caption{Comparison of AREL with QMIX and a strategic exploration technique, MAVEN in the $3s\_ vs \_ 5z$ StarCraft map (avg. over 5 runs). AREL yields highest rewards and win rates.}\label{Explorn}
\end{figure}
\begin{figure*}[!ht]
	\includegraphics[width=0.55\linewidth]{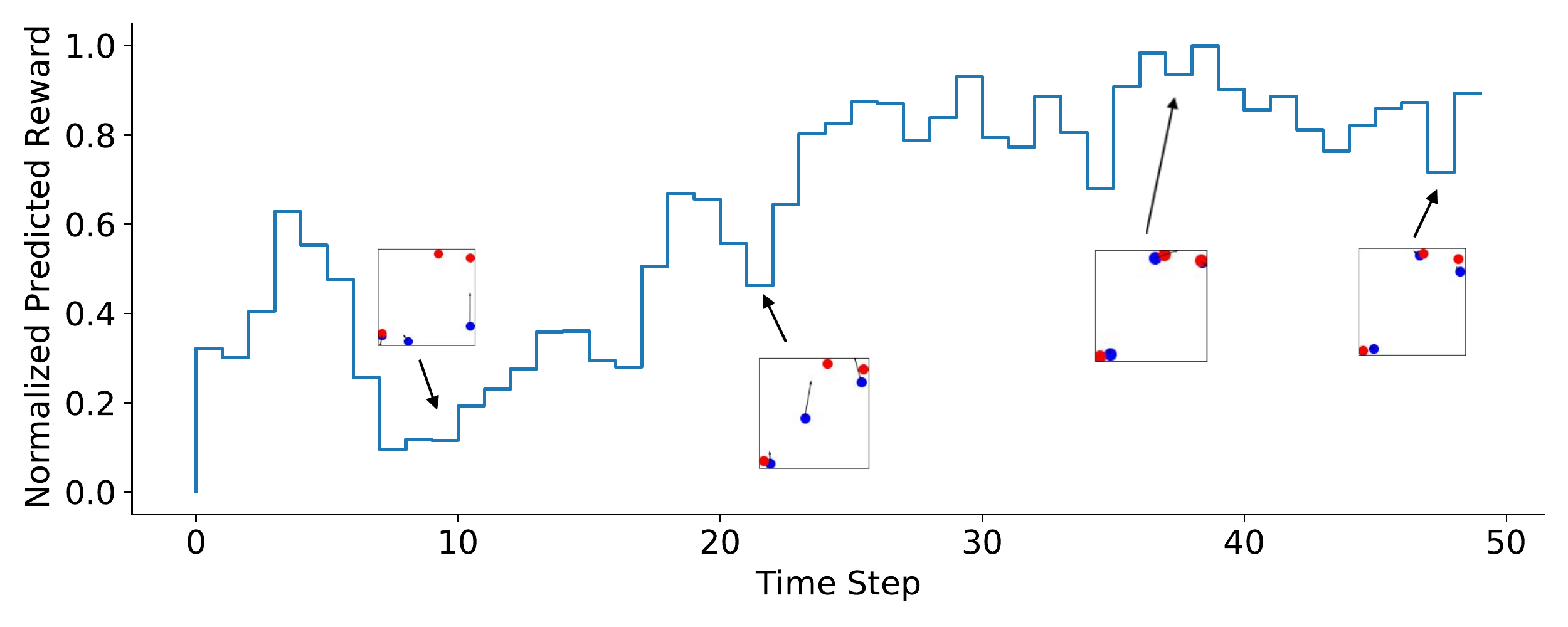}
	\caption{An instantiation of the Cooperative Navigation task with $N=3$ where rewards are provided only at the end of an episode. Blue and red dots respectively denote agents and landmarks. Arrows on agents represent their directions of movement. The objective of this task is for each agent to cover a distinct landmark. The $y-$axis of the graph shows the $0-1$ normalized predicted rewards for a sample trajectory. The positions of agents relative to landmarks are shown at several points along this trajectory. The figure shows a scenario where two agents are close to a single landmark. In this case, one of them must remain close to this landmark, while the other moves towards a different landmark. The predicted redistributed reward encourages such an action, since it has a higher magnitude when agents navigate towards distinct landmarks. The predicted redistributed reward by \emph{AREL} is not uniform along the length of the episode.} \label{Pred_Rew}
\end{figure*}
\subsubsection{\textbf{Credit Assignment vs. exploration}}
%\textcolor{blue}{
This section demonstrates the importance of effective redistribution of an episodic reward \emph{vis-a-vis} strategic exploration of the environment. 
The episodic reward $R_T (= \sum_t r_t)$ takes continuous values and provides fine-grained information on performance (beyond only win/ loss). 
AREL learns a redistribution of $R_T$ by identifying critical states in an episode, and does not provide exploration abilities beyond that of the base RL algorithm. 
The redistributed rewards of AREL can be given as input to any RL algorithm to learn policies (in our experiments, we demonstrate using QMIX for StarCraft; MADDPG for Particle World).  
Figure \ref{Explorn} illustrates a comparison of \emph{AREL} with a state-of-the-art exploration strategy, MAVEN \cite{mahajan2019maven} and with QMIX \cite{rashid2018qmix} in the $3s\_ vs \_ 5z$ StarCraft map. 
We observe that when rewards are delayed to the end of an episode, effectively redistributing the reward can be more beneficial than strategically exploring the environment to improve win-rates or total rewards. 
%}
\subsubsection{\textbf{Interpretability of Learned Rewards}}
Figure \ref{Pred_Rew} presents an interpretation of the decomposed predicted rewards \emph{vis-a-vis} the relative positions of agents to landmarks in the Cooperative Navigation task with $N = 3$. 
When the reward is provided only at the end of an episode, \emph{AREL} is used to learn a temporal redistribution of this episodic reward. 
The predicted rewards are normalized to a $0-1$ scale for ease of representation. 
The positions of the agents relative to the landmarks are shown at several points along a sample trajectory. 
Successfully trained agents must learn policies that enable each agent to cover a distinct landmark. 
For example, in a scenario where two agents are close to a single landmark, one of them must remain close to this landmark, while the other moves towards a different landmark. 
%This insight is revealed in the positions of the agents at various steps along the trajectory. 
We observe that the magnitude of the predicted rewards is consistent with this insight in that it is higher when agents navigate away and towards different landmarks. 

%\textcolor{blue}{
This visualization in Figure \ref{Pred_Rew} reveals that the attention mechanism in \emph{AREL} is able to learn to redistribute an episodic reward effectively in order to successfully train agents to accomplish task objectives in cooperative multi-agent reinforcement learning. 
Moreover, it reveals that the redistributed reward predicted by \emph{AREL} is not uniform along the length of the episode. 

\subsection{Discussion}

This paper focused on developing techniques to effectively learn policies in MARL environments when rewards were delayed or episodic. 
Our experiments demonstrate that \emph{AREL} can be used as a module that enables more effective credit assignment by identifying critical states through capturing long-term temporal dependencies between states and an episodic reward. 
Redistributed rewards predicted by \emph{AREL} are dense, which can then be provided as an input to MARL algorithms that learn value functions for credit assignment (we used MADDPG \cite{lowe2017multi} and QMIX \cite{rashid2018qmix} in our experiments). 

By including a variance-based regularization term, the total loss in Eqn. (\ref{LossTotal}) enabled incorporating the possibility that not all intermediate states would contribute equally to an episodic reward, while also learning less sparse redistributed rewards. 
Moreover, any exploration ability available to the agents was provided solely by the MARL algorithm, and not by \emph{AREL}. 
We further demonstrated that effective credit assignment was more beneficial than strategic exploration of the environment when rewards are episodic. 

\section{Conclusion}
%\textcolor{red}{
%This paper presented a solution approach for long-term temporal credit assignment in multi-agent tasks with episodic rewards.
%Solving such MARL problems require addressing two challenges: 
%identifying (1) relative importance of states along the length of an episode (along time), and (2) relative importance of individual agent's state at any single time-step (among agents).
%To deal with these two challenges, attention-based method, \emph{AREL}, were used to characterize the effect of actions on state transitions along trajectories (temporal attention), and how each agent was influenced by other agents at each time-step (agent attention).
%To improve sample efficiency, potential agent permutation invariance among homogeneous agents is considered.
%The temporally redistribution of episodic reward predicted by AREL offers dense supervision and could be integrated with any given MARL algorithms.
%AREL was evaluated on tasks from the Particle World environment and the StarCraft Multi-Agent Challenge, and compared with three state-of-the-art reward redistribution techniques. 
%\emph{AREL} resulted in agents obtaining higher rewards in Particle World and better win rates in StarCraft.}
%
This paper studied the multi-agent temporal credit assignment problem in MARL tasks with episodic rewards. 
Solving this problem required addressing the twin challenges of identifying the relative importance of states along the length of an episode and individual agent's state at any single time-step. 
We presented an attention-based method called \emph{AREL} to deal with the above challenges. 
%\emph{AREL} used a property of permutation invariance among homogeneous agents to improve sample efficiency. 
The temporally redistributed reward predicted by \emph{AREL} was dense, and could be integrated with MARL algorithms. 
%\emph{AREL} was evaluated on tasks from Particle World and StarCraft, and was compared with three state-of-the-art reward redistribution techniques. 
%\emph{AREL} resulted in agents obtaining higher rewards in Particle World and better win rates in StarCraft. 
\emph{AREL} was evaluated on tasks from Particle World and StarCraft, and was successful in obtaining higher rewards and better win rates than three state-of-the-art reward redistribution techniques. 
%
%
%We note that our experiments require computational resources to train the attention modules of \emph{AREL} in addition to those needed to train deep RL algorithms. 
%Using these resources might result in higher energy consumption, especially as the number of agents grows. 
%This is a potential limitation of the methods studied in this paper. 
%However, we believe that \emph{AREL} partially addresses this concern by sharing certain modules among all agents in order to improve scalability.
%
\section*{Acknowledgment}

This work was supported by the Office of Naval Research via Grant N00014-17-S-B001.
\bibliographystyle{ACM-Reference-Format} 
\bibliography{CredAssRL}

%%%%%%%%%%%%%%%%%%%%%%%%%%%%%%%%%%%%%%%%%%%%%%%%%%%%%%%%%%%%%%%%%%%%%%%%
\clearpage
%\onecolumn
%\begin{center}\textbf{\LARGE Agent-Temporal Attention for Reward Redistribution in\\ Episodic Multi-Agent Reinforcement Learning} \end{center}

%\begin{center}\textbf{\Large Paper ID: zZyY} \end{center}

\section{Appendices}

These appendices include detailed analysis and proofs of the theoretical results in the main paper. 
They also contain details of the environments and present additional experimental results and ablation studies.

\subsection*{\large Appendix A: Analysis}

Using the variance of the predicted redistributed rewards as the regularization loss allows us to provide an interpretation of the overall loss in Eqn. (\ref{LossTotal}) in terms of a bias-variance trade-off in a mean square estimator. 
The variance-regularized predicted rewards are analyzed in detail in the sequel. 
%We provide an interpretation of the overall loss in Eqn. (\ref{LossTotal}) in terms of a bias-variance trade-off in a mean square estimator. 

First, assume that the episodic reward $R_T$ is given by the sum of `ground-truth' rewards $r_t$ at each time step. 
That is, $R_T = \sum_t r_t$. 
The objective in Eqn. (\ref{LossTotal}) is:%can then be written as: 
\begin{align}
\arg \min_{\theta} loss_{total}(\theta)&= \arg \min_{\theta} (l_r(\theta) + \omega l_v (\theta)). \label{LossTotalObj}
\end{align}
This type of regularization will determine $ f_{arel}^{\theta}$ in a manner such that it will be \emph{robust to overfitting}. 

For a sample trajectory, the expectation over $R_T$ in $l_r(\theta)$ can be omitted. 
Let $\bar{f}^{\theta}_t:= \mathbb{E}_{\mathbf{E}}  (f_{arel}^{\theta}( \mathbf{E}_t))$ (i.e., the mean of the predicted rewards). 
Moreover, let $\bar{f}^{\theta}_t = \bar{f}^{\theta} (\mathbf{E}):= \frac{1}{T}\sum_t (f_{arel}^{\theta}( \mathbf{E}_t))$ (i.e., $\mathbf{E}_t$ is an ergodic process). 
Then, the following holds:
\begin{align}
loss_{total}(\theta)
&=\mathbb{E}_{\mathbf{E}} \big[\frac{1}{T}\big(
\sum_t (f_{arel}^{\theta}( \mathbf{E}_t ) - r_t ) \big)^2 \nonumber \\ 
 & \qquad \qquad + \frac{\omega}{T}\sum_t \big(f_{arel}^{\theta}( \mathbf{E}_t ) - \bar{f}^{\theta}_t \big)^2 \big] \label{LossTotExpanded}\\
&\leq \mathbb{E}_{\mathbf{E}} \big[ \frac{1}{T}
T\sum_t \big(f_{arel}^{\theta}( \mathbf{E}_t ) - r_t \big)^2 \nonumber \\ 
& \qquad \qquad + \frac{\omega}{T}\sum_t \big(f_{arel}^{\theta}( \mathbf{E}_t ) - \bar{f}^{\theta}_t \big)^2 \big] \label{CauchySchwarzIneq}
\end{align}

The first term in (\ref{CauchySchwarzIneq}) is obtained by applying the Cauchy-Schwarz inequality to the first term of (\ref{LossTotExpanded}). 
Consider $\mathbb{E}_{\mathbf{E}} [\sum_t \big(f_{arel}^{\theta}( \mathbf{E}_t ) - r_t \big)^2]$. 
From linearity of the expectation operator, this is equal to $\sum_t[\mathbb{E}_{\mathbf{E}} \big(f_{arel}^{\theta}( \mathbf{E}_t ) - r_t \big)^2]$. 
Then, $f_{arel}^{\theta}( \mathbf{E}_t )$ can be interpreted as an estimator of $r_t$, and the expression above is the mean square error of this estimator. 
By adding and subtracting $\bar{f}^{\theta}_t$, the mean square error can be expressed as the sum of the variance and the squared bias of the estimator \cite{domingos2000unified}. 
Formally, 
\begin{align*}
&\mathbb{E}_{\mathbf{E}} [\big(f_{arel}^{\theta}( \mathbf{E}_t ) - r_t \big)^2]=\mathbb{E}_{\mathbf{E}} [\big(f_{arel}^{\theta}( \mathbf{E}_t ) - \bar{f}^{\theta}_t+\bar{f}^{\theta}_t-r_t \big)^2]\nonumber \\
&=\mathbb{E}_{\mathbf{E}}[(f_{arel}^{\theta}( \mathbf{E}_t ) - \bar{f}^{\theta}_t)^2+(\bar{f}^{\theta}_t-r_t)^2%\nonumber \\&\qquad \qquad \qquad
+2(f_{arel}^{\theta}( \mathbf{E}_t ) - \bar{f}^{\theta}_t)(\bar{f}^{\theta}_t-r_t)].\nonumber
\end{align*}
After distributing the expectation in the above expression, we obtain the following: \\
\textbf{(a)}: since $(\bar{f}^{\theta}_t-r_t)$ is constant, the third term $\mathbb{E}_{\mathbf{E}}[(f_{arel}^{\theta}( \mathbf{E}_t ) - \bar{f}^{\theta}_t)(\bar{f}^{\theta}_t-r_t)]$ is equal to      
%$\mathbb{E}_{\mathbf{E}}[(f_{arel}^{\theta}( \mathbf{E}_t ) - \bar{f}^{\theta}_t)(\bar{f}^{\theta}_t-r_t)] = 
$(\bar{f}^{\theta}_t-r_t)\mathbb{E}_{\mathbf{E}}[(f_{arel}^{\theta}( \mathbf{E}_t ) - \bar{f}^{\theta}_t)] = (\bar{f}^{\theta}_t-r_t) (\bar{f}^{\theta}_t - \bar{f}^{\theta}_t) = 0$; \\
\textbf{(b)}: the first two terms correspond to the variance of a random variable $f_{arel}^{\theta}( \mathbf{E}_t )$ and the square of a bias between $\bar{f}^{\theta}_t$ and $r_t$, and may not be both zero.\\
%The first term in the above equation is the variance of $f_{arel}^{\theta}( \mathbf{E}_t )$, the second term is the square of the bias, and the third term evaluates to zero. 
Substituting these in Eqn. (\ref{CauchySchwarzIneq}),% we obtain: 
\begin{align}
loss_{total}(\theta) &\leq \sum_t \big[(1+\frac{\omega}{T})\mathbb{E}_{\mathbf{E}}[(f_{arel}^{\theta}( \mathbf{E}_t ) - \bar{f}^{\theta}_t)^2]\nonumber \\
&\qquad \qquad +\mathbb{E}_{\mathbf{E}}[(\bar{f}^{\theta}_t-r_t)^2]\big]. \label{LossTotalBound}
\end{align}

Therefore, the total loss is upper-bounded by an expression that represents the sum of a bias and a variance. 
The parameter $\omega$ will determine the relative significance of each term. 
Let $loss_{rhs}(\theta)$ represent the term on the right hand side of (\ref{LossTotalBound}). 
If we denote by $\theta_1$ the parameters that minimize $loss_{total}(\theta)$, and by $\theta_2$ the parameters that minimize $loss_{rhs}(\theta)$, then an optimization carried out on $loss_{total}(\theta)$ can be related to one carried out on $loss_{rhs}(\theta)$ as: 
\begin{align*}
loss_{total}(\theta_1) \leq loss_{total}(\theta_2) \leq loss_{rhs}(\theta_2) \leq loss_{rhs}(\theta_1).
\end{align*}
%
%\begin{rem} 
%An implicit assumption made here is that $r_t \equiv r_t^{\theta}$ at each time-step. This is required to ensure that the expectation operator in the above set of inequalities is valid. 
%\end{rem}

For the special case when the mean of the predicted rewards, $\bar{f}^{\theta} = \frac{R_T}{T}$, the first term of $loss_{total}(\theta)$ in Eqn. (\ref{LossTotExpanded}) will evaluate to zero. 
The optimization of $loss_{total}(\theta)$ in this case is then reduced to minimizing the square error of predictors $f^\theta_{arel}(E_t)$ at each time $t \in \{0,1,\dots,T-1\}$. 
This setting is consistent with the principle of maximum entropy when the objective is to distribute an episodic return uniformly along the length of the trajectory. 

Consider a single time-step $t_1$, and assume that there are enough samples (say, $S_{max}$) to `learn' $f^\theta_{arel}(E_{t_1})$. 
Then, at each time-step $t_1$, the goal is to solve the problem $\mathbb{E}_{\theta} [(f^\theta_{arel}(E_{t_1}) - \frac{R_T}{T})^2]$. 

The squared loss above will admit a bias-variance decomposition \cite{domingos2000unified, geman1992neural} that is commonly interpreted as a trade-off between the two terms. 
This underscores an insight that the complexity of the estimator (in terms of dimension of the set containing the parameters $\theta$) should achieve an `optimal' balance \cite{friedman2001elements, goodfellow2016deep}. 
This is represented as a U-shaped curve for the total error, where bias decreases and variance increases with the complexity of the estimator. 
However, recent work has demonstrated that the variance of the prediction also decreases with the complexity of the estimator \cite{belkin2019reconciling, neal2018modern}. 

%We present a bound on the error of the variance of the predictor at each time-step. 
In order to determine a bound on the error of the variance of the predictor at each time-step, we first state a result from \cite{neal2018modern}. 
We use this to provide a bound on $loss_{total}(\theta)$ when the estimators $f^\theta_{arel}(E_t)$ are unbiased at each time-step in Theorem \ref{ThmBoundLoss}. 

\begin{thm}\cite{neal2018modern}\label{ThmBiasVar}
Let $N$ be the dimension of the parameter space containing $\theta$. 
Assume that the parameters $\theta$ are initialized by a Gaussian as $\theta_0 \sim \mathcal{N}(0, \frac{1}{N}I)$. 
Let $L_{t_1}=o(\sqrt{N})$ be a Lipschitz constant associated to $f^\theta_{arel}(E_{t_1})$. 
Then, for some constant $C>0$, the variance of the prediction satisfies $Var_{\theta_0}(f^\theta_{arel}(E_{t_1})) \leq 2CL_{t_1}^2/N$. 
\end{thm}

\begin{thm}\label{ThmBoundLoss}
Let $\bar{f}^{\theta} := \mathbb{E}_{\mathbf{E}}  (f_{arel}^{\theta}( \mathbf{E}))= R_T/T$ denote the mean of the predicted rewards. Assume that there are $S_{max}$ samples to `learn' $f^\theta_{arel}(E_{t_1})$ at each time-step $t_1 \in \{0,1,\dots,T-1\}$. Then, for unbiased estimators $f^\theta_{arel}(E_{t_1})$ and associated Lipschitz constants $L_{t_1}=o(\sqrt{N})$,% at each $t_1$, 
\begin{align*}
loss_{total}(\theta) &\leq \frac{2 \omega C}{NT} \sum_t (L_0^2 + \dots L_{T-1}^2).
\end{align*}
\end{thm}
%
%
%Let $\bar{f}^{\theta}:= \mathbb{E}_{\mathbf{E}}  (f_{arel}^{\theta}( \mathbf{E}))$ denote the mean of the predicted rewards and assume 
%%Consider the case %when the mean of the predicted rewards, 
%$\bar{f}^{\theta} = R_T/T$. 
%The more general case is analyzed in \emph{Appendix B}. 
%In this setting, the first term of $loss_{total}(\theta)$ in Eqn. (\ref{LossTotExpanded}) will evaluate to zero. 
%Theorem \ref{ThmBoundLoss} indicates that the optimization of $loss_{total}(\theta)$ will be equivalent to minimizing the square error of predictors $f^\theta_{arel}(E_t)$ at each time $t \in \{0,1,\dots,T-1\}$. 
%This setting is consistent with the principle of maximum entropy when we are interested in distributing an episodic return uniformly along the length of the trajectory. 
%
%At a single time-step $t_1$, assume that we have enough %there is a sufficient number of 
%samples (say, $S_{max}$) to `learn' $f^\theta_{arel}(E_{t_1})$.  
%Then, the goal is to solve the problem $\mathbb{E}_{\theta} [(f^\theta_{arel}(E_{t_1}) - \frac{R_T}{T})^2]$ at each $t_1$. 
%This setting is consistent with the principle of maximum entropy when the objective is to distribute an episodic return uniformly along the length of the trajectory.

\begin{proof}%[\textbf{Theorem \ref{ThmBoundLoss}}]  
The proof follows from applying Theorem \ref{ThmBiasVar} at each time step $t_1 \in \{0,1,\dots,T-1\}$. 
Since the estimators at each time-step are unbiased, the mean square error of predictors $f^\theta_{arel}(E_t)$ is equal to the variance of the predictor. 
The expectation operator in Eqn. (\ref{LossTotExpanded}) is now over a constant, which completes the proof. 
\end{proof}

The assumption on estimators being unbiased is reasonable. 
Proposition \ref{PropUnb} indicates that in the more general case (i.e. for an estimator that may not be unbiased), the prediction is concentrated around its mean with high probability. 

\begin{prop}\cite{wainwright2019high} \label{PropUnb}
Under a Gaussian initialization of parameters as $\theta_0 \sim \mathcal{N}(0, \frac{1}{N}I)$, at each time-step $t_1$, the following holds:
\begin{align*}
\mathbb{P}(|f^\theta_{arel}(E_{t_1}) - \mathbb{E}_\theta[f^\theta_{arel}(E_{t_1})]| > \epsilon) \leq 2 ~exp(-\frac{CN\epsilon^2}{L^2}).
\end{align*}
\end{prop}

Theorem \ref{ThmBoundLoss} indicates that the optimization of $loss_{total}(\theta)$ will be equivalent to minimizing the square error of predictors $f^\theta_{arel}(E_t)$ at each time $t \in \{0,1,\dots,T-1\}$. 

\subsection*{\large Appendix B: Detailed Task Descriptions} \label{TaskDescriptions}

This Appendix gives a detailed description of the tasks that we evaluate \emph{AREL} on. % in this Appendix. 
In each experiment, a reward is obtained by the agents only at the end of an episode. 
No reward is provided at other time steps. 
%The delayed %nature of the 
%reward makes it difficult for agents to receive immediate feedback on the quality of their actions at each time step. 

\begin{itemize}
\item \textbf{Cooperative Push}: This task has $N$ agents working together to move a large ball to a landmark. Agents are rewarded when the ball reaches the landmark. Each agent observes its position and velocity, relative position of the target landmark and the large ball, and relative positions of the $k$ nearest agents. We report results for $(N, k) = (3, 2)$, $(6, 5)$, and $(15, 10)$. 
At each time step, the distance $d^t_{a,b}$ between agents and the ball, distance $d^t_{b,l}$ between ball and landmark, and whether the agents touch the ball $I^t$ is recorded. These quantities, though, will be not be immediately revealed to the agents. 
Agents receive a reward $\sum_{t=1}^T(-\lambda_1 d^t_{a,b} - \lambda_2 d^t_{b,l} + \lambda_3 I^t)$ at the end of each episode at time $T$. 

\item \textbf{Predator-Prey}: This task has $N$ predators working together to capture $M$ preys. $L$ landmarks impede movement of the agents. Preys can move faster than predators, and predators obtain a positive reward when they collide with a prey. The $M$ prey agents are controlled by the environment. Each predator observes its position and velocity, relative locations of the $l$ nearest landmarks, and relative positions and velocities of the $k_1$ nearest prey and $k_2$ nearest predators. We report results for $(N, M, L, k_1,  k_2, l) = (3, 1, 2, 1, 3, 2)$, $(6, 2, 3, 2, 6, 3)$, and $(15, 5, 5, 3, 6, 3)$. 
At each time step, the distance $d^t_{prey,pred}$ between a prey and the closest predator, and whether a predator touches a prey $I^t$ is recorded. These quantities, though, will be not be immediately revealed to the agents. 
The agents receive a reward $\sum_{t=1}^T(-\lambda_1 d^t_{prey,pred} + \lambda_2 I^t)$ at the end of each episode at time $T$.

\item \textbf{Cooperative Navigation}: This task has $N$ agents seeking to reach $N$ landmarks. The maximum reward is obtained when there is exactly one agent at each landmark. Agents are also penalized for colliding with each other. Each agent observes its position, velocity, and the relative locations of the $k$ nearest landmarks and agents. We report results for $(N,k) = (3, 2), (6, 5), (15, 5)$. 
At each time step, the distance $d^t_{a,l}$ between an agent and the closest landmark, and whether an agent collides with other agents $I^t$ is recorded. These quantities, though, will be not be immediately revealed to the agents. 
The agents receive a reward $\sum_{t=1}^T(-\lambda_1 d^t_{a,l} - \lambda_2 I^t)$ at the end of each episode at time $T$.

\item \textbf{StarCraft}: We use the SMAC benchmark from \cite{samvelyan19smac}. 
The environment comprises two groups of army units, and units from one group (controlled by learning agents) collaborate to attack units from the other (controlled by handcrafted heuristics).
%Units from the same group collaborate to attack units from the other group. 
%One group is controlled by handcrafted heuristics and the other is controlled by our learning agents.
Each learning agent controls one army unit. 
We report results for three maps: 2 Stalkers and 3 Zealots (2s3z); 1 Colossus, 3 Stalkers, and 5 Zealots (1c3s5z); and 3 Stalkers versus 5 Zealots (3s\_vs\_5z).
In 2s3z and 1c3s5z, two groups of identical units are placed
symmetrically on the map. 
In 3s\_vs\_5z, the learning agents control 3 Stalkers to attack 5 Zealots controlled by the StarCraft AI.
In all maps, units can only observe other units if they are both alive and located within the \emph{sight range}. 
The 2s3z and 1c3s5z maps comprise heterogeneous agents, since there are different types of units, while the 3s\_vs\_5z is a homogeneous map. 
%The reward structures for the maps in StarCraft are more complex. 
In all our experiments, the default environment reward is delayed and revealed only at the end of an episode. The reader is referred to \cite{samvelyan19smac} for a detailed description of the default rewards. 
\end{itemize}

\subsection*{\large Appendix C: Implementation Details} \label{ImplementationDetails}
\begin{figure*}
	\begin{subfigure}{0.31\textwidth}
		\includegraphics[width=\linewidth]{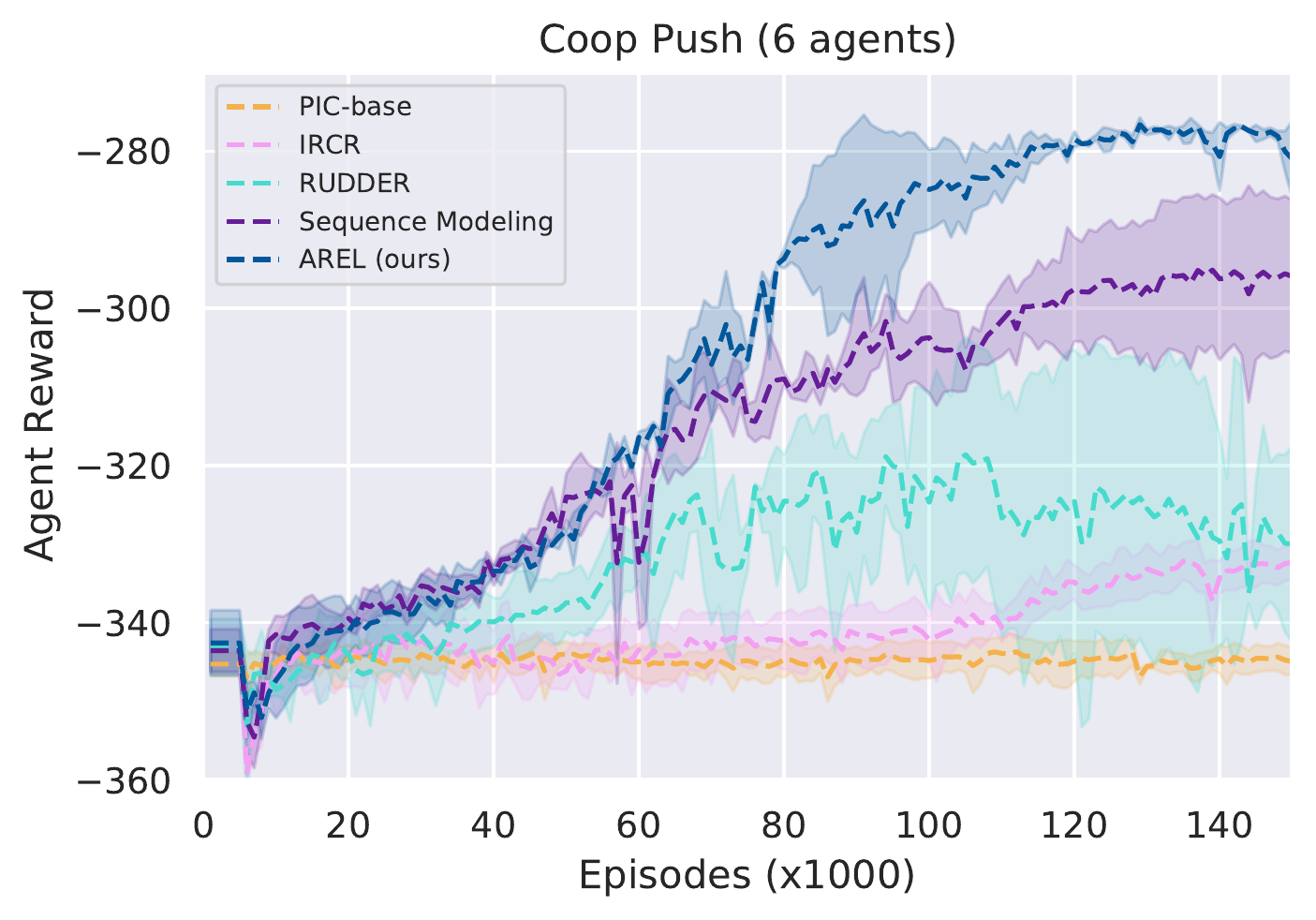}
		\caption{Cooperative Push} \label{fig:1a}
	\end{subfigure}%
	\hspace*{\fill}   % maximize separation between the subfigures
	\begin{subfigure}{0.31\textwidth}
		\includegraphics[width=\linewidth]{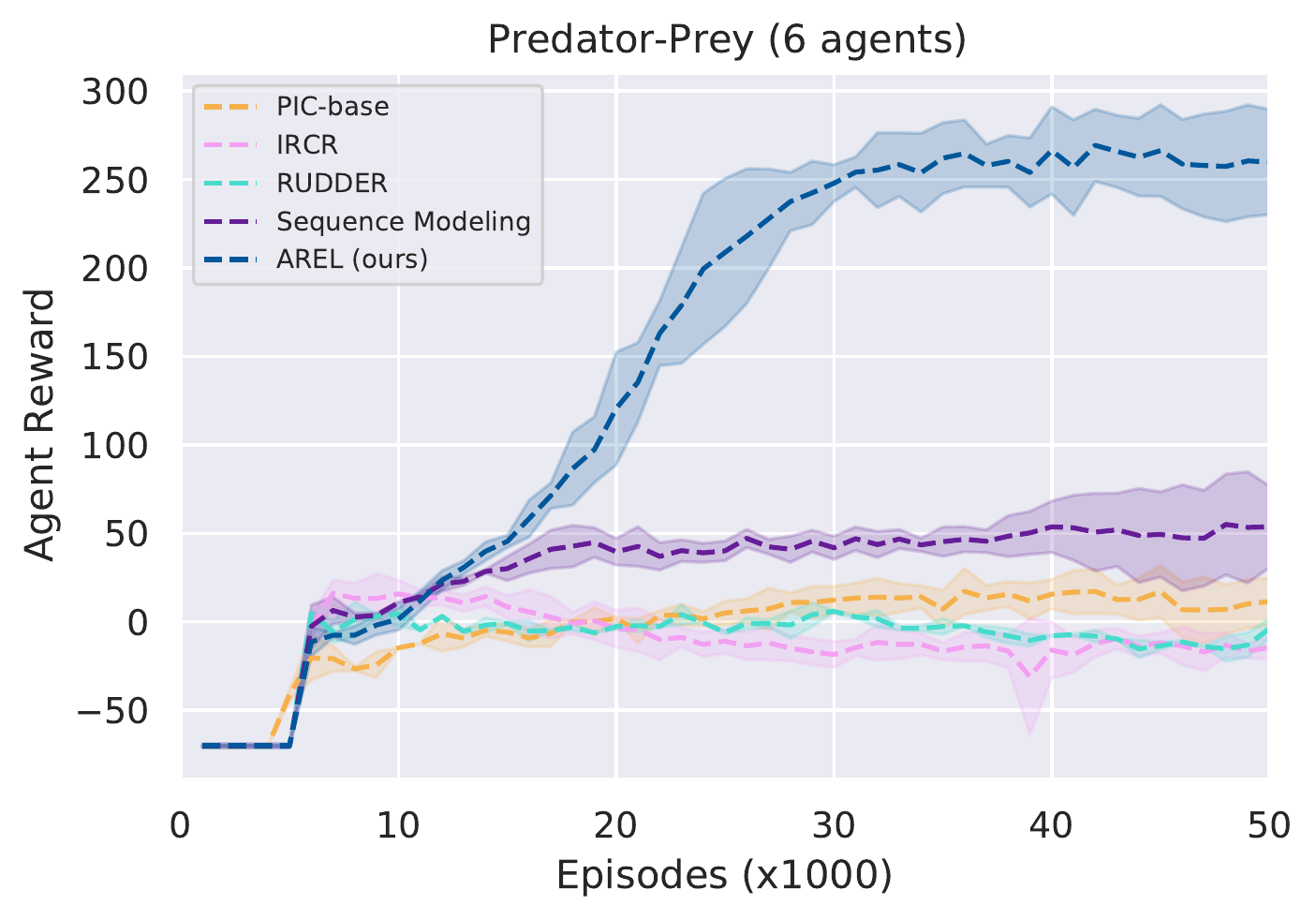}
		\caption{Predator and Prey} \label{fig:1b}
	\end{subfigure}%
	\hspace*{\fill}   % maximizeseparation between the subfigures
	\begin{subfigure}{0.31\textwidth}
		\includegraphics[width=\linewidth]{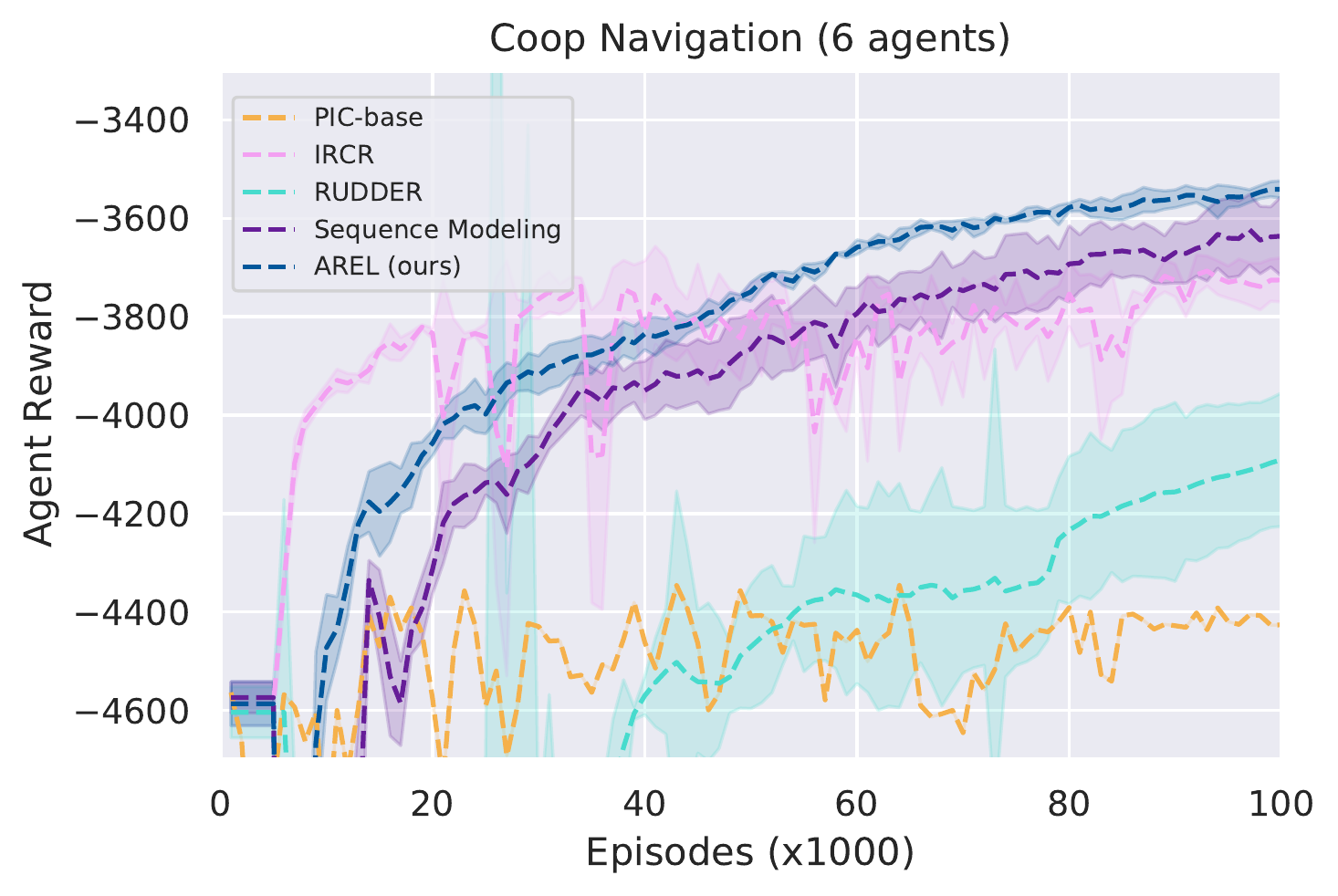}
		\caption{Cooperative Navigation} \label{fig:1c}
	\end{subfigure}
	
	\caption{Average agent rewards and standard deviation for tasks in the Particle World environment with episodic rewards and $N=6$. \emph{AREL} (dark blue) results in the highest average rewards in all the tasks. %In comparison, the PIC baseline (orange) and RUDDER (green) are unable to learn policies to complete the tasks. Using a surrogate objective in Trajectory Smoothing (pink) results in comparable rewards during the early stages of training in Cooperative Navigation, but the reward curve is unstable. Sequence modeling (purple) obtains a higher reward than the former three methods, but does not explicitly model agent-attention, which results in a lower average reward than \emph{AREL}.
	} \label{FigGraphsN6}
\end{figure*}

\begin{figure*}
	\begin{subfigure}{0.31\textwidth}
		\includegraphics[width=\linewidth]{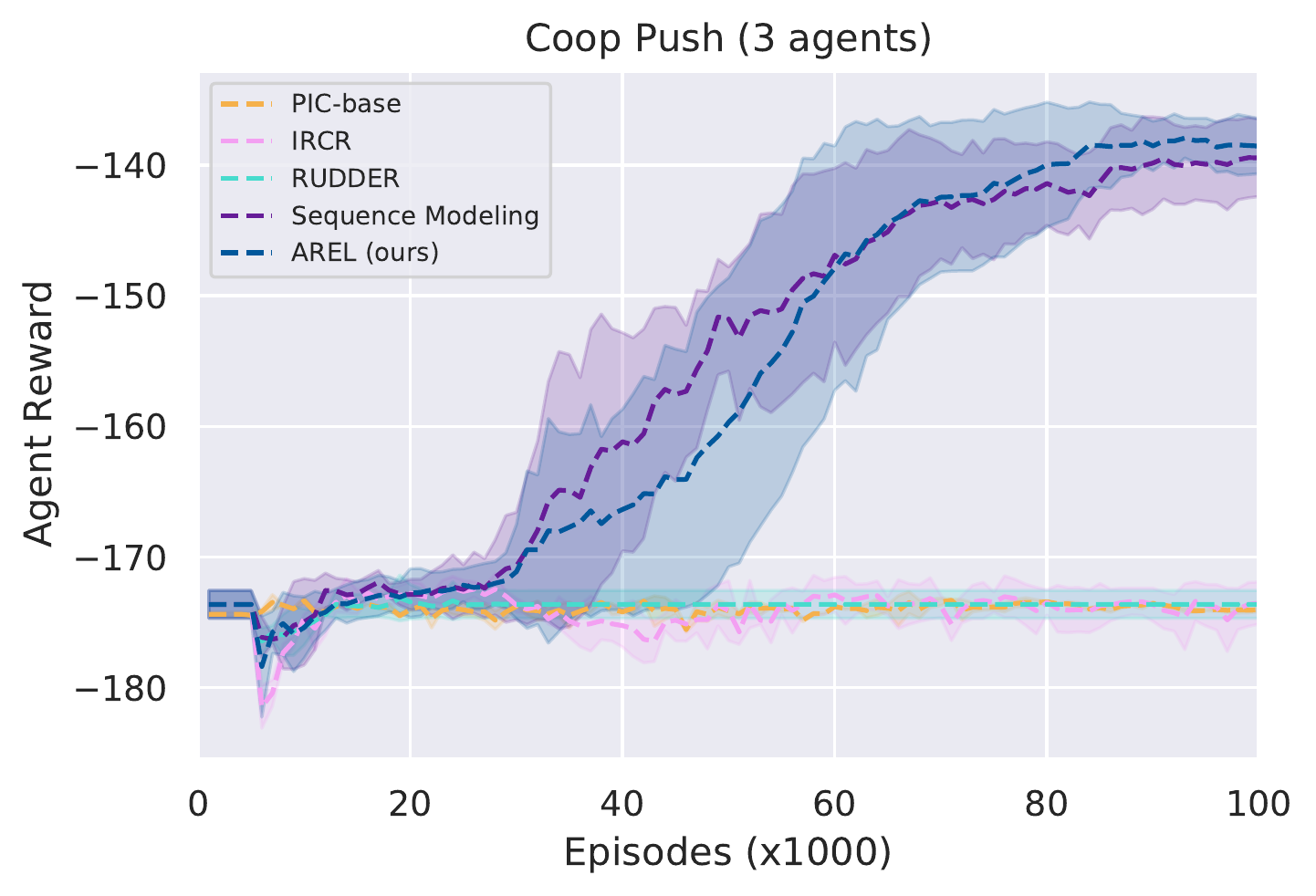}
		\caption{Cooperative Push} \label{fig:3a}
	\end{subfigure}%
	\hspace*{\fill}   % maximize separation between the subfigures
	\begin{subfigure}{0.31\textwidth}
		\includegraphics[width=\linewidth]{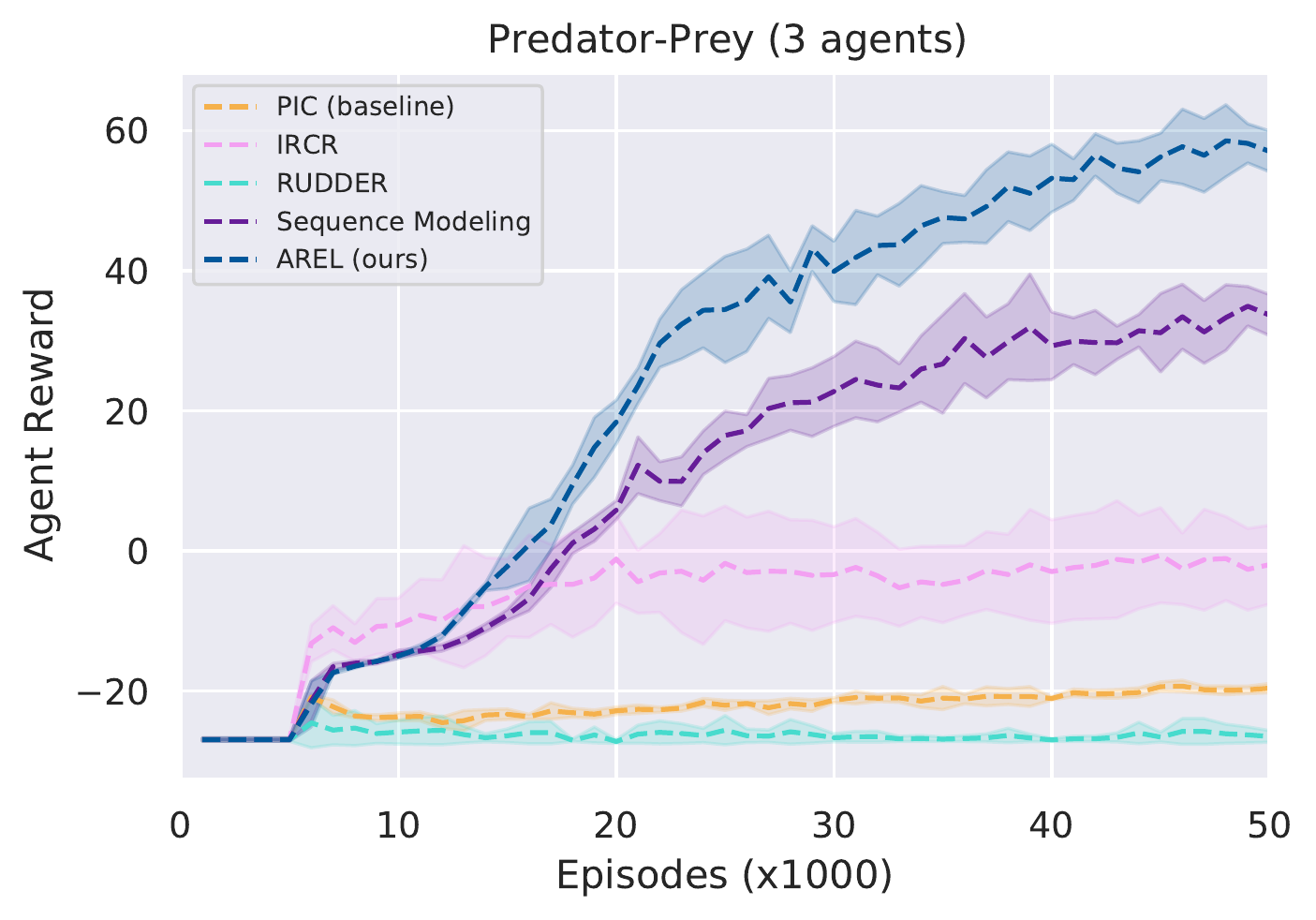}
		\caption{Predator and Prey} \label{fig:3b}
	\end{subfigure}%
	\hspace*{\fill}   % maximizeseparation between the subfigures
	\begin{subfigure}{0.31\textwidth}
		\includegraphics[width=\linewidth]{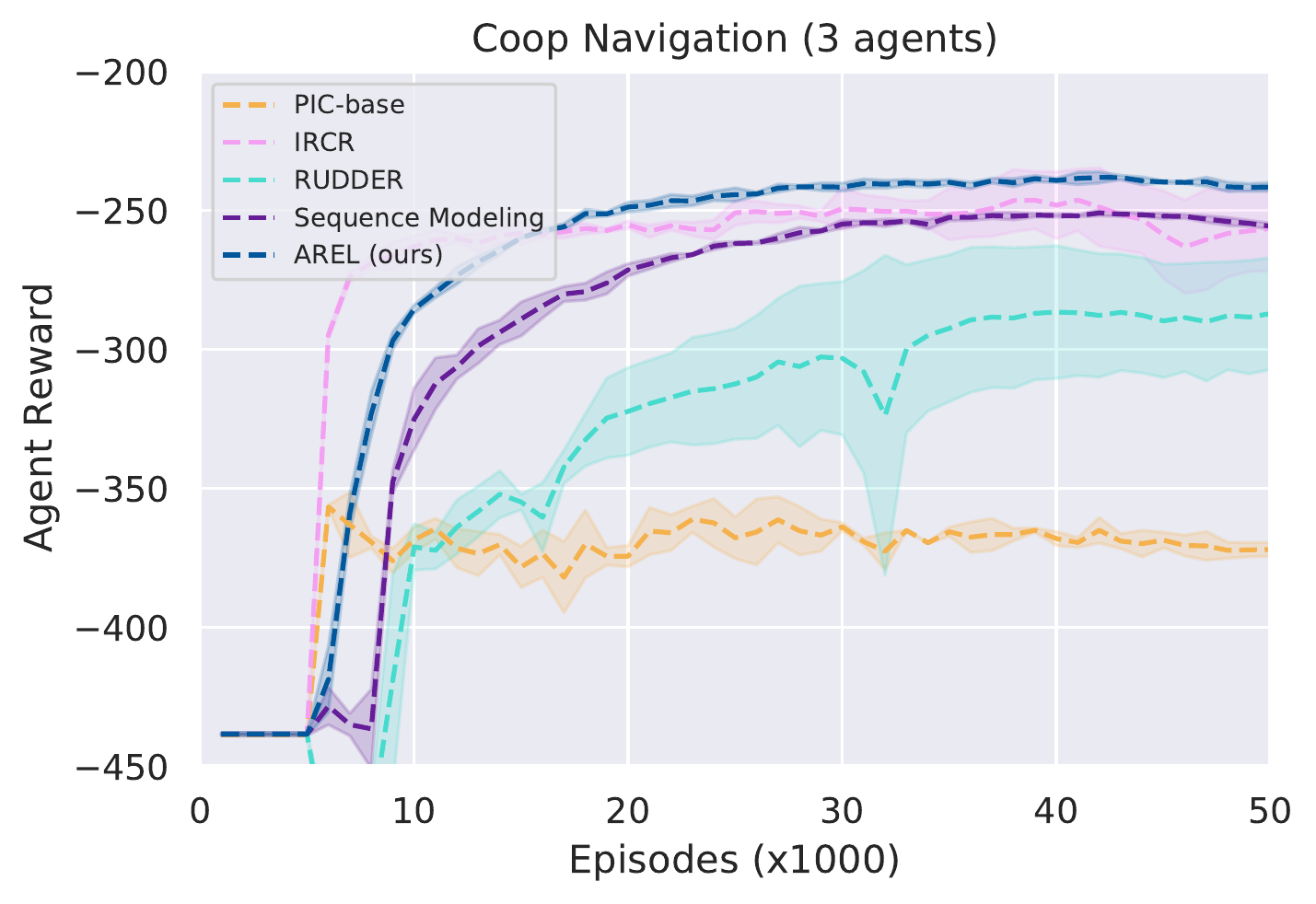}
		\caption{Cooperative Navigation} \label{fig:3c}
	\end{subfigure}
	
	\caption{Average agent rewards in Particle World environment with episodic rewards and $N=3$. \emph{AREL} (dark blue) results in the highest average rewards. %In comparison, the PIC baseline (orange) and RUDDER (green) are unable to learn policies to complete the tasks. Using a surrogate objective in Trajectory Smoothing (pink) results in comparable rewards in Cooperative Navigation. The performance of Sequence modeling (purple) is comparable to \emph{AREL} in this setting, which indicates that the effect of modeling agent attention is less pronounced when there are fewer agents.
	} \label{FigGraphsN3}
\end{figure*}

All the results presented in this paper are averaged over 5 runs with different random seeds. 
We tested the following values for the regularization parameter in Eqn. (\ref{LossTotal}): $\{0.1, 1, 10, 20, 50, 100\}$, and observed that $\omega = 20$ resulted in the best performance. 

%The regularization parameter $\omega$ in Equation \ref{LossTotal} is set to 20, which is searched among $\{0.1, 1, 10, 20, 50, 100\}$.
In order to make the agent-temporal attention module more expressive, we use a transformer architecture with multi-head attention \cite{vaswani2017attention} for both agent and temporal attention.
Specifically, in our experiments, the transformer architecture applies, in sequence: an attention layer, layer normalization, two feed forward layers with ReLU activation, and another layer normalization.
Before each layer normalization, residual connections are added.

When dimension of the observation space exceeds $100$, a single fully connected layer with $100$ units is applied to compress the observation before attention module. 
The credit assignment block that produces redistributed rewards $\hat{r}_t$ consists of two MLPs, each with a single hidden layer of 50 units.
Neural networks for credit assignment are trained using Adam with learning rate $1\times 10^{-4}$. 

In Particle World, credit assignment networks are updated for $1000$ batches every $1000$ episodes. 
Each batch contains $256$ fully unrolled episodes uniformly sampled from the trajectory experience buffer $B_e$.  
In StarCraft, credit assignment networks are updated for $400$ batches every $200$ episodes. %, the credit assignment networks are updated for 400 batches. 
Each batch contains $32$ fully unrolled episodes uniformly sampled from $B_e$. 

During training and testing, the length of each episode in Particle World is kept fixed at $25$ time steps, except in Predator-Prey with $N=15$, where the episode length is set to $50$ time steps.
In StarCraft, an episode is restricted to have a maximum length of $120$ time steps for 2s3z, $180$ time steps
for 1c3s5z, and $250$ time steps for 3s\_vs\_5z. 
If both armies are alive at the end of the episode, we count it as a loss for the team of learning agents. 
An episode terminates after one army has been defeated, or the time limit has been reached.

In Particle World, we use the permutation invariant critic (\emph{PIC}) based on MADDPG from \cite{liu2020pic} as the base reinforcement learning algorithm. 
The code is based on the implementation available at \hyperlink{https://github.com/IouJenLiu/PIC}{https://github.com/IouJenLiu/PIC}.
Following MADDPG \cite{lowe2017multi}, the actor policy is parameterized by a two-layer MLP with 128 hidden units per layer, and ReLU activation function. 
The permutation invariant critic is a two-layer graph convolution net with 128 hidden units per layer, max pooling at the top, and ReLU activation. 
%The replay buffer consists of the last 1000000 steps, from which a batch of size 1024 is uniformly sampled for training via Adam optimizer.
Learning rates for actor and critic are 0.01, and is linearly decreased to zero at the end of training. 
Trajectories of the first $5000$ episodes are sampled randomly for filling the experience buffer.
During training, uniform noise was added for exploration during action selection.

In StarCraft, we use QMIX \cite{rashid2018qmix} as the base algorithm. 
The QMIX code is based on the implementation from \hyperlink{https://github.com/starry-sky6688/StarCraft}{https://github.com/starry-sky6688/StarCraft}. 
In the implementation, all agent networks share a deep recurrent Q-network with recurrent layer comprised of a GRU with a 64-dimensional hidden
state, with a fully-connected layer before and after.  
%The replay buffer consists of the last 5000 episodes, from which a batch of size 32 is uniformly sampled for training.
Trajectories of the first $20000$ episodes are sampled randomly to fill the experience buffer.
Target networks are updated every 200 training episodes.
QMIX is trained using RMSprop with learning rate $5 \times 10^{-4}$.
Throughout training, $\epsilon$-greedy is adopted for exploration, and  
$\epsilon$ is annealed linearly from 1.0 to 0.05 over 50k time steps and kept constant for the rest of learning.

The experiments that we perform in this paper require computational resources to train the attention modules of \emph{AREL} in addition to those needed to train deep RL algorithms. 
Using these resources might result in higher energy consumption, especially as the number of agents grows. 
This is a potential limitation of the methods studied in this paper. 
However, we believe that \emph{AREL} partially addresses this concern by sharing certain modules among all agents in order to improve scalability.  

We provide a description of our hardware resources below: \\
\textbf{Hardware Configuration}: All our experiments were carried out on a machine running \emph{Ubuntu}\textsuperscript \textregistered $18.04$ equipped with a $16-$core \emph{Intel}\textsuperscript \textregistered \emph{Xeon}\textsuperscript \textregistered $3.7$ GHz CPU, two \emph{NVIDIA}\textsuperscript \textregistered \emph{GeFORCE}\textsuperscript \textregistered RTX $2080$ Ti graphics cards and a $128$ GB RAM.

\subsection*{\large Appendix D: Additional Experimental Results}  \label{AddnlExptRes}
\begin{table*}[]
	\begin{tabular}{|c|c|c|c|c|c|}
		\hline
		\multicolumn{1}{|c|}{\textbf{No. of Agents}} & \textbf{Task}          & \textbf{$R_{avg}$: AREL} & \textbf{$R_{avg}$: Uniform} & \textbf{$R_{final}$: AREL} & \textbf{$R_{final}$: Uniform} \\ \hline
		& CP       &   $\mathbf{-1490.7}$                       &     -1553.6                        & $\mathbf{-1375.2}$                         &        $-1477.6$                     \\ \cline{2-6} 
		N = 15                                       & PP          &    $\mathbf{827.7}$                      & $674.0$                            &   $\mathbf{1325.0}$                       &          $1167.0$                   \\ \cline{2-6} 
		& CN &         $\mathbf{-2691.1}$                 &    $-2811.2$                         &    $\mathbf{-2206.2}$                      &     $-2301.1$                        \\ \hline 
		& CP       & $\mathbf{-306.9}$                         &   $-317.6$                          & $\mathbf{-282.4}$                         &     $-286.3$                        \\ \cline{2-6} 
		N = 6                                        & PP          &    $\mathbf{149.4}$                      & $111.6$                            &$\mathbf{259.8}$                          & $229.7$                            \\ \cline{2-6} 
		& CN &    $\mathbf{-3902.3}$                      &     $-3950$                        &     $\mathbf{-3541.2}$                     &$-3648.6$                             \\ \hline
		& CP       &    $\mathbf{-159.0}$                      &           $-166.5$                  &  $\mathbf{-139.2}$                        &         $-146.1$                    \\ \cline{2-6} 
		N = 3                                        & PP          &   $\mathbf{23.4}$                       &           $15.2$                  &    $\mathbf{57.0}$                      &          $46.9$                   \\ \cline{2-6} 
		& CN &  $\mathbf{-274.5}$                 &       $-282.4$                      &         $\mathbf{-241.5}$                 &     $-244.4$                        \\ \hline 
	\end{tabular}

\caption{Ablation: This table demonstrates the effect of removing the agent attention module, and uniformly weighting the attention of each agent in the three tasks in Particle World. This is termed \emph{Uniform}, and we compare this with \emph{AREL}. We report the average rewards over the number of training episodes ($R_{avg}$) and the final agent reward ($R_{final}$) at the end of training in both scenarios for each task. \emph{AREL} consistently results in higher average and maximum rewards (shown in \textbf{bold}), which indicates that the agent attention module plays a crucial role in effective credit assignment.}
\label{tab:my-table}
\end{table*}

This Appendix presents additional experimental results carried out in the Particle World environment. 

Figure \ref{FigGraphsN6} shows results of our experiments for tasks in Particle World when $N = 6$. 
In each case, \emph{AREL} is consistently able to allow agents to learn policies that result in higher average rewards compared to other methods. 
This is a consequence of using an attention mechanism that enables decomposition of an episodic reward along the length of an episode, and that also characterizes contributions of individual agents to the reward. 
Performances of the \emph{PIC} baseline \cite{liu2020pic}, \emph{RUDDER} \cite{arjona2019rudder}, and \emph{Sequence Modeling} \cite{liu2019sequence} can be explained similar to that presented for the case $N=15$ in the main paper. 
Using a surrogate objective in \emph{IRCR} \cite{gangwani2020learning} results in obtaining comparable agent rewards in some cases in the Cooperative Navigation task, but the reward curves are unstable and have high variance. 
%
%In comparison, the \emph{PIC} baseline fails to learn policies to accomplish the tasks when rewards are provided only at the end of an episode. 
%We observe that \emph{RUDDER} is also unable to learn agent policies to complete these tasks. 
%An explanation for this could be that RUDDER only carries out a temporal decomposition of rewards, but does not take into consideration the effect of different agents contributing differently to a reward. 
%
%\emph{Sequence Modeling} performs better than former two methods, possibly because it uses a Transformer-based attention mechanism for the temporal decomposition of an episodic reward. 
%This was shown to outperform LSTM-based models in \cite{liu2019sequence} in single-agent episodic RL tasks, due to the relative ease of training the Transformer-based mechanism. 
%However, the absence of an explicit characterization of agent-attention is a possible justification for the smaller average reward of this method compared to \emph{AREL}. 
%
%Using a surrogate objective in \emph{IRCR} results in obtaining comparable agent rewards in some cases in the Cooperative Navigation task, but the reward curves are unstable and have a high variance. 
%%\emph{Smooth} also obtains an average reward comparable to \emph{AREL} in the Predator-Prey task with $N = 15$. 
%However, \emph{AREL} consistently obtains higher average rewards than \emph{IRCR} during the early stages of training. 
%The performance of \emph{IRCR} is comparable to \emph{RUDDER} and the \emph{PIC} baseline in all other tasks. 

Figure \ref{FigGraphsN3} show the results of experiments on these tasks when $N = 3$. 
The \emph{PIC} baseline and \emph{RUDDER} are unable to learn good policies and \emph{IRCR} results in lower rewards than \emph{AREL} in two tasks. 
The performance of \emph{Sequence Modeling} is comparable to \emph{AREL}, which indicates that characterizing agent attention plays a smaller role when there are fewer agents. 

\subsection*{\large Appendix E: Additional Ablations} \label{AddnlAblations} 
%\begin{figure*}[ht]
%\centering
%\includegraphics[width=0.51\linewidth]{figures/reg_loss.pdf}
%\caption{Ablation: Effect of value of regularization parameter $\omega$ in Eqn (\ref{LossTotal}) on magnitude of regression loss $l_r$.}\label{Diff_w_Loss_Ratio}
%\end{figure*}

\begin{figure}[ht]
\centering
\includegraphics[width=0.85\linewidth]{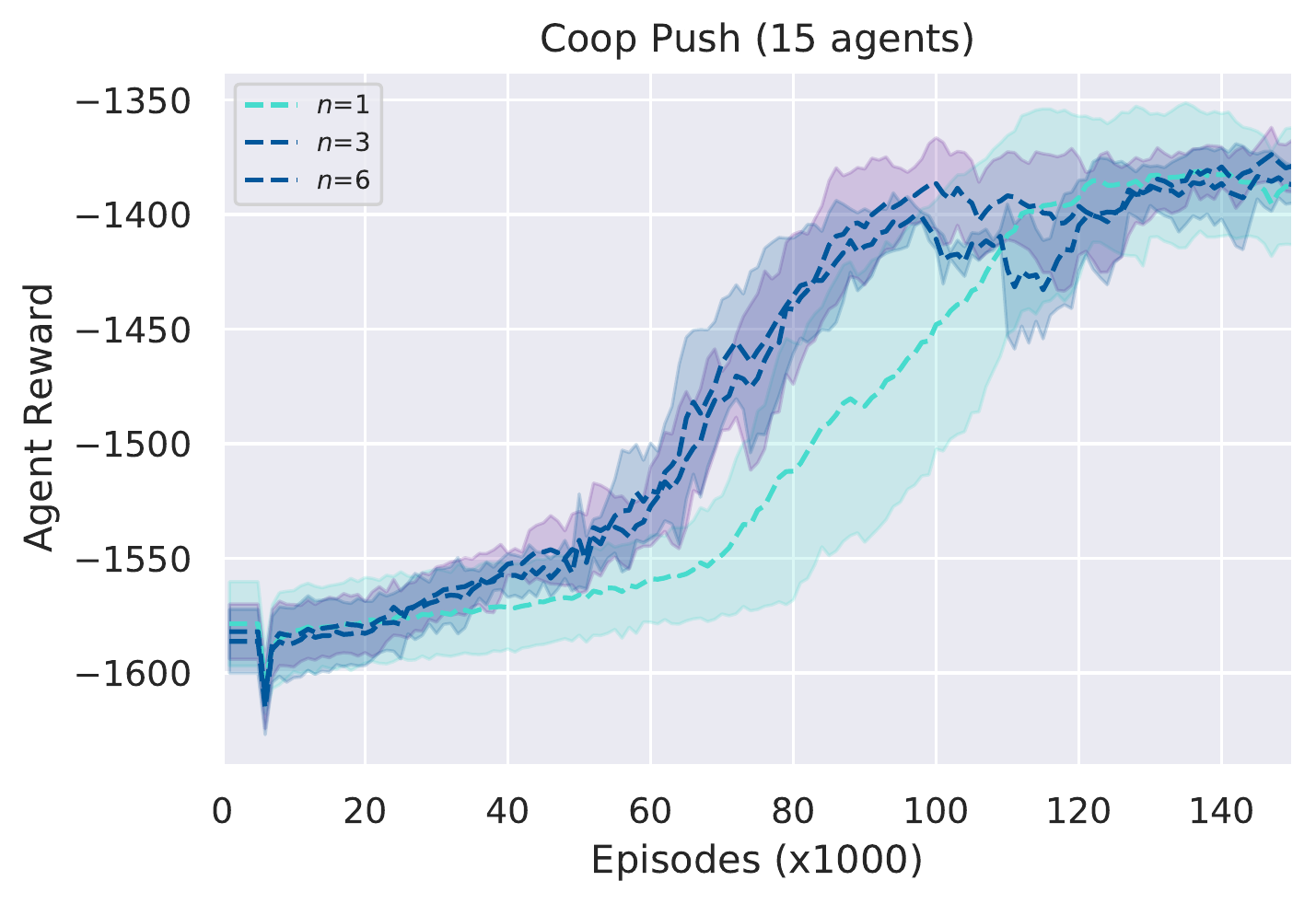}
\caption{Ablation: Effect of the number of agent-temporal attention blocks on rewards.}\label{Diff_w}
\end{figure}

This Appendix presents additional ablations that examine the impact of uniformly weighting the attention of agents %, the value of the coefficient $\omega$ in Eqn. (\ref{LossTotal}) on the magnitude of the regression loss $l_r$, 
%$loss_{total}(\theta)$ (Eqn. (\ref{LossTotal})), 
and the number of agent-temporal attention blocks.%, and the number of agent-temporal attention blocks on agent rewards. 

We evaluate the effect of removing the agent-attention block, and uniformly weighting the attention of each agent. 
This is termed uniform agent attention (\emph{Uniform}). 
The average rewards over the number of training episodes ($R_{avg}$) and the final agent reward ($R_{final}$) at the end of training obtained when using \emph{AREL} and when using \emph{Uniform} are compared. 
The results of these experiments, presented in Table \ref{tab:my-table}, indicate that $R_{avg}$ and $R_{max}$ are higher for \emph{AREL} than for \emph{Uniform}. 
This shows that the agent-attention block in \emph{AREL} plays a crucial role in performing credit assignment effectively.

We examine the effect of the number of agent-temporal attention blocks, $n$ (\textbf{depth}) on rewards in the Cooperative Push task with $N=15$ in Figure \ref{Diff_w}. 
The depth has negligible impact on average rewards at the end of training. 
However, rewards during the early stages of training are lower for $n=1$, and these rewards also have a larger variance than the other cases ($n=3, 6$). 

\subsection*{\large Appendix F: Effect of Choice of Regularization Loss}
%\textcolor{blue}{
This Appendix examines the effect of the choice of the regularization loss term in Eqn. (\ref{LossTotal}). 
The need for the regularization loss term arises due to the possibility that there could be more than one choice of redistributed rewards that minimize the regression loss alone. 
In our results in the main paper, we used the variance of the redistributed rewards as the regularization loss. 
This choice was motivated by a need to discourage the predicted redistributed rewards from being sparse, since sparse rewards might impede learning of policies when provided as an input to a MARL algorithm \cite{devlin2011theoretical}. 
By adding a variance-based regularization, the total loss enables incorporating the possibility that not all intermediate states would contribute equally to an episodic reward, while also resulting in learning redistributed rewards that are less sparse. 
%Our motivation for this choice was to devise a way to generalize the solution of the \emph{IRCR} algorithm presented in \cite{gangwani2020learning}, where the authors uniformly redistribute an episodic reward along the length of the episode. 
%}

%\textcolor{blue}{
%With respect to Eqn. (\ref{LossTotal}), the parameter $\omega$ taking a very high value can be interpreted as an approximation of \emph{IRCR} \cite{gangwani2020learning}, where the authors uniformly redistributed an episodic reward along the length of the episode. 
%An exact uniform redistribution of the predicted reward might not be possible (please see Discussion in Sec. 4.2.2). %if one is additionally interested in using information from sample trajectories to characterize the relative contributions of agents at different intermediate time-steps along an episode. 
%%However, in the presence of additional information that might be available about the regression loss, 
%This insight is additionally underscored by observations that very large values of $\omega$ results in smaller agent rewards (Fig. 4(b)), and that the redistributed rewards are indeed not uniform (visualization in \emph{Appendix I}). 
%}

%\textcolor{blue}{
We compare the variance-based regularization loss with two other widely used choices of the regularization loss- the $L_1$ and $L_2$-based losses. 
The $L_1$-based regularization encourages learning sparse redistributed rewards, and the $L_2$-based regularization discourages learning a redistributed reward of large magnitude (i.e., `spikes' in the redistributed reward). 
Specifically, we study:
%} 
\begin{align*}
loss_{total}(\theta) &= l_r (\theta) + \omega l_v (\theta)\\
loss_{total}(\theta) &= l_r (\theta) + \omega_1 l_1 (\theta)\\
loss_{total}(\theta) &= l_r (\theta) + \omega_2 l_2 (\theta), 
\end{align*}
where $l_r (\theta), l_v(\theta)$ are the regression loss %as in Eqn. (\ref{LossTotal}), 
and variance of redistributed rewards as in Eqn. (\ref{LossTotal}), $l_1 (\theta)$ ($l_2 (\theta)$) is the $L_1$ norm ($L_2$ norm) of the redistributed reward.%, and $l_2 (\theta)$ is the $L_2$ norm of the redistributed reward. 

%\textcolor{blue}{
We compare the use of the three regularization loss functions on the tasks in \emph{Particle World} with $N = 15$. 
In each task, we calculate the $0-1$-normalized reward received by the agents during the last $1000$ training steps. 
We use $\omega = 20$ for the variance-based regularization loss. 
For the other two regularization losses, we searched over $\omega_1 \in \{0.1, 1, 10, 100 \}$ and $\omega_2 \in \{0.1, 1, 10, 100 \}$, and we observed that $\omega_1 = 10$ and $\omega_2 = 10$ resulted in the best performance. 
The graph in Figure \ref{FigRegLosses} shows the average $0-1$-normalized reward. 
We observe that using a variance-based regularization loss results in agents obtaining the highest average rewards.
%}
\begin{figure}[ht]
\centering
\includegraphics[width=0.85\linewidth]{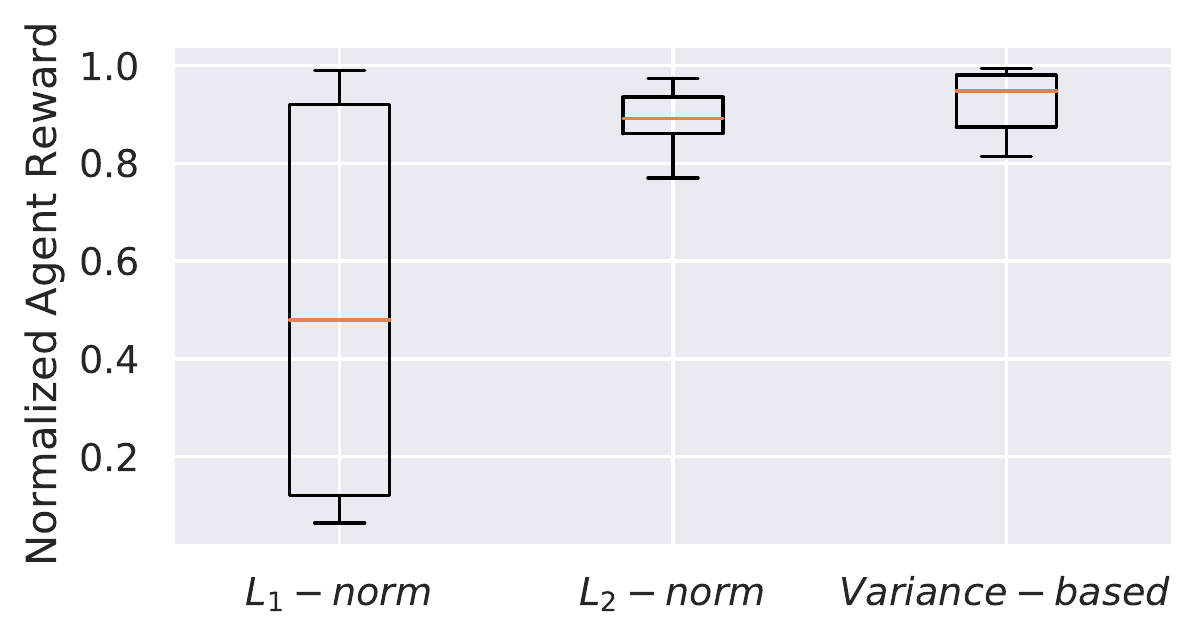}
\caption{Normalized average agent rewards for tasks in \emph{Particle World} when using variance-based, $L_1$-based, and $L_2$-based regularization losses. The variance-based regularization results in agents obtaining the highest average rewards.}\label{FigRegLosses}
\end{figure}

%\textcolor{blue}{
In particular, we observe that using an $L_1$-based regularization results in significantly smaller rewards. 
A possible reason for this is that the $L_1$-based regularization has the property of encouraging learning a \emph{sparse} redistributed reward, which hinders learning of policies when provided as an input to the MARL algorithm. 
The performance when using the $L_2$-based regularization results in a comparable, albeit slightly lower, average agent reward to using the variance-based regularization. 
This is reasonable since using the variance-based or $L_2$-based regularization will result in less sparse predicted redistributed rewards. 
%}

\subsection*{\large Appendix G: Verification of QMIX Implementation}

\begin{figure}[ht]
\centering
\includegraphics[width=0.85\linewidth]{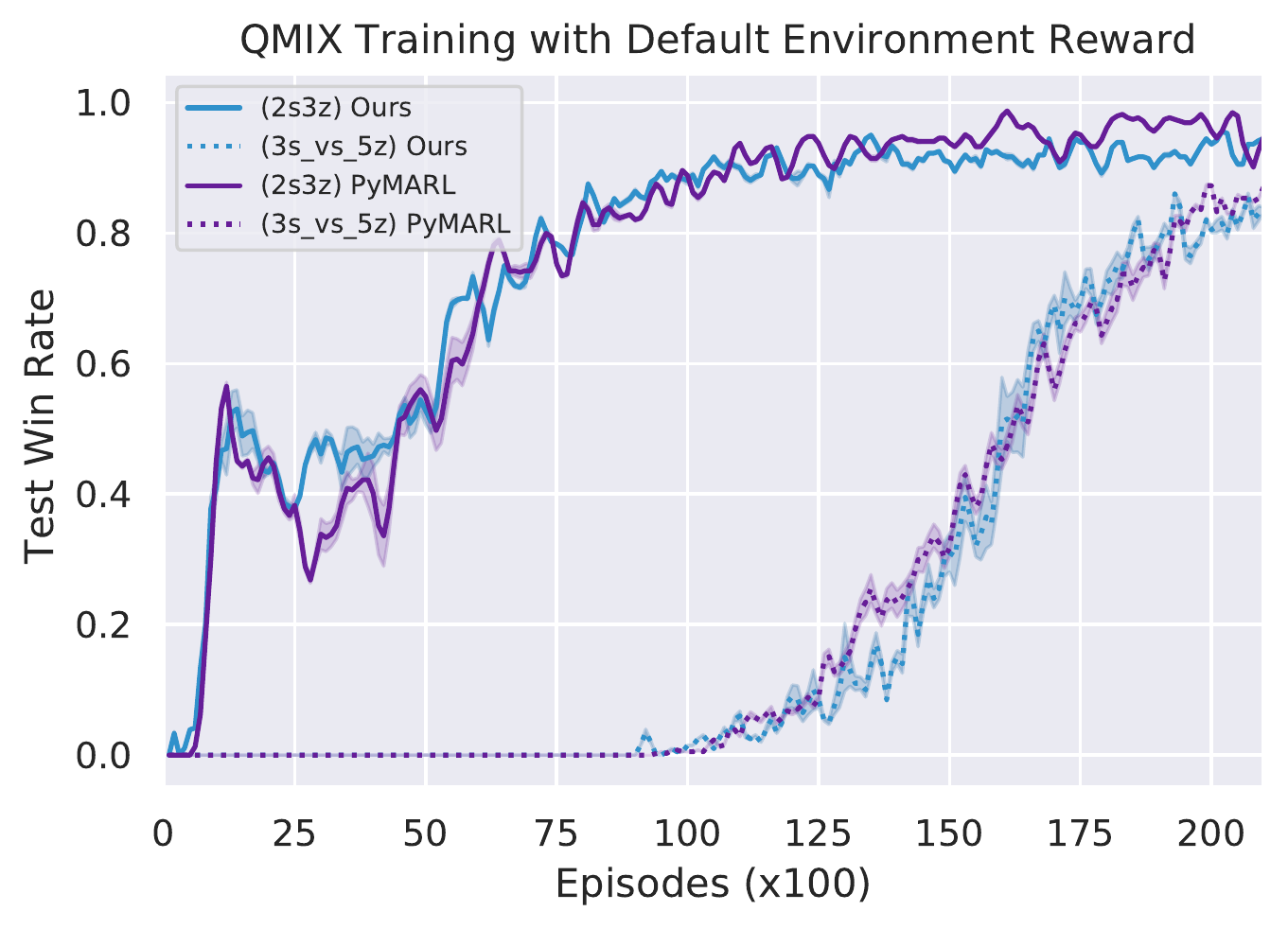}
\caption{Comparison of PyMARL and our implementations (avg. over 5 runs). Test win-rates using either implementation are almost identical.}\label{PyMARLvsOurs}
\end{figure}

%\textcolor{blue}{
This Appendix demonstrates correctness of the QMIX implementation that we use from \hyperlink{https://github.com/starry-sky6688/StarCraft}{https://github.com/starry-sky6688/StarCraft}. 
In the QMIX evaluation first used in \cite{rashid2018qmix}, rewards were not delayed. 
In our experiments, rewards are delayed and revealed only at the end of an episode. 
In such a scenario, QMIX may not be able to perform long-term credit assignment, which explains the difference in performance between the default and delayed reward cases. 
We observe that using redistributed rewards from AREL as an input to QMIX results in an improved performance compared to using QMIX alone when rewards from the environment were delayed (Figure 3 in the main paper). 
Using the default, non-delayed rewards, we compare the performance of the QMIX implementation that we used in our experiments with the benchmark implementation from \cite{samvelyan19smac}. 
Figure \ref{PyMARLvsOurs} shows that test win rates in two StarCraft maps ($2s3z$ and $3s\_vs\_5z$) using both implementations are almost identical. 
%}

\end{document}